\documentclass[lettersize,journal]{IEEEtran}

\usepackage{tikz}
\usepackage{amsmath}
\usepackage{ulem}
\usepackage{graphicx}
\usepackage{epstopdf}
\usepackage{xspace}
\usepackage{pifont}
\usepackage{multirow}
\usepackage{adjustbox}
\usepackage{booktabs} 
\usepackage{comment}
\usepackage{float}
\usepackage{subfigure}
\usepackage{url} 
\usepackage{algorithm}
\usepackage{algorithmic}
\usepackage{caption}
\usepackage{color}
\usepackage{listings}
\usepackage{xcolor}
\usepackage{booktabs}
\usepackage{makecell}
\usepackage{multirow}
\newcommand{\bheading}[1]{{\vspace{4pt}\noindent{\textbf{#1}}}}
\newcommand{\iheading}[1]{{\vspace{2pt}\noindent{\textit{#1}}}}

\newcounter{note}[section]

\newcommand{\secref}[1]{\mbox{Sec.~\ref{#1}}\xspace}

\newcommand{\figref}[1]{\mbox{Fig.~\ref{#1}}}

\newcommand{\ignore}[1]{}

\newcommand{\ie}{\textit{i.e.}\xspace}
\newcommand{\eg}{\textit{e.g.}\xspace}

\newcommand{\etal}{\textit{et al.}\xspace}

\newcommand{\sysname}{\textsc{Odyssey}\xspace}
\newcommand{\attack}{\text{execution-based}\xspace}
\newcommand{\eiatt}{\text{execution-inference}\xspace}
\newcommand{\eratt}{\text{execution-replay}\xspace}

\newcommand{\EIatt}{\text{Execution-Inference}\xspace}
\newcommand{\ERatt}{\text{Execution-Replay}\xspace}

\newcommand{\blackding}[1]{\ding{\numexpr181+#1\relax}}

\newcounter{packednmbr}

\newenvironment{packeditemize}{
\begin{list}{$\bullet$}{
\setlength{\labelwidth}{0pt}
\setlength{\itemsep}{2pt}
\setlength{\leftmargin}{\labelwidth}
\addtolength{\leftmargin}{\labelsep}
\setlength{\parindent}{0pt}
\setlength{\listparindent}{\parindent}
\setlength{\parsep}{1pt}
\setlength{\topsep}{1pt}}}{\end{list}}

\newtheorem{theorem}{Theorem}
\newtheorem{lemma}{Lemma}

\newtheorem{definition}{Definition}

\newtheorem{ob}{Observation}

\def\BibTeX{{\rm B\kern-.05em{\sc i\kern-.025em b}\kern-.08em
    T\kern-.1667em\lower.7ex\hbox{E}\kern-.125emX}}
\usepackage{balance}
\begin{document}

\def\thetitle{\sysname: Reestablishing Confidentiality in
Confidential Blockchain via Delegated Execution}
\title{\thetitle}

\author{
Ju Yang, 
Weili Wang,  
Jianyu Niu,~\IEEEmembership{Member, ~IEEE}, 
Jianzong Wang, 
Yinqian Zhang~\IEEEmembership{Member, ~IEEE}

\IEEEcompsocitemizethanks{
        \IEEEcompsocthanksitem Ju Yang, Weili Wang, Jianyu Niu, and Yinqian Zhang are with the Research Institute of Trustworthy Autonomous Systems and the Department of Computer Science and Engineering, Southern University of Science and Technology, Shenzhen, China.
        Email: 12431264@mail.sustech.edu.cn, 12032870@mail.sustech.edu.cn, niujy@sustech.edu.cn, and yinqianz@acm.org.

        \IEEEcompsocthanksitem Jianzong Wang is with Ping An Technology (Shenzhen) Co., Ltd., Shenzhen, China. Email: jzwang@188.com. 
    }
    
}

\maketitle

\begin{abstract}
Confidential blockchains leveraging Trusted Execution Environments (TEEs) have garnered extensive attention for transaction confidentiality.
In this paper, we first taxonomize two classes of attacks against confidential blockchains, \ie, \eiatt and \eratt attacks, which exploit TEEs' long-lasting side-channel and state-continuity issues to compromise the confidentiality of existing consortium blockchains. 
Then, we present \sysname, a confidential blockchain that efficiently mitigates these attacks.
The core innovations of \sysname are the following: (1) Its delegation model: clients delegate transaction execution to their designated trustees, while other participants synchronize only the execution results, which significantly reduces the attack surface while preserving confidentiality and system performance. 
(2) Two novel techniques to improve \sysname's efficiency and security: \textit{location-aware concurrent execution} and \textit{delegation failure handler}.   
Finally, we develop a prototype of \sysname on FISCO BCOS, an enterprise-grade consortium blockchain platform. 
We have conducted various experiments, and our evaluation results show that in a WAN environment with 3 nodes, \sysname can achieve about 4k throughput while keeping latency as low as 0.4-0.5s. 
\end{abstract}


\section{Introduction}
\label{sec:intro}
\IEEEPARstart{I}{n} 2008, Nakamoto invented blockchain technology, which enables a set of mutually untrusting nodes to agree on a forever-growing sequence of transactions, to realize the first decentralized currency, Bitcoin~\cite{nakamoto2008bitcoin}. 
Later, many permissionless blockchain platforms like Ethereum~\cite{pos_eth} adopt the same open model (\ie, not knowing identities of nodes) as Bitcoin, while Consortium blockchains like Hyperledger Fabric~\cite{androulaki2018hyperledger}, Quorum~\cite{quorum}, BCOS~\cite{li2023fisco}, and Corda~\cite{brown2016corda} support the permissioned model. The permissioned model allows a set of organizations that wish to keep data confidential to construct a consortium and maintain a blockchain for storing private data and running distributed applications (also known as smart contracts) like cross-border exchange~\cite{islam2022low}, goods tracing~\cite{assaqty2020private}, and supply chain~\cite{wu2023high}.

Despite the permissioned features, some consortium blockchains processing financial data or supply chain data require a higher level of confidentiality, requiring that on-chain information is not exposed to consortium members.
To meet the demands, confidential blockchains~\cite{russinovich2019ccf, yan2020confidentiality, howard2023ccf} utilize Trusted Execution Environments (TEEs) to provide confidentiality for data and applications in consortium blockchains. 
TEEs are CPU extensions that leverage hardware security features, such as state isolation, hardware-assisted memory encryption, and remote attestation, to protect code and data without trusting the underlying privileged system software.
Specifically, consortium nodes run distributed applications in TEEs to process transactions. 
All transactions are encrypted before being sent to the confidential blockchain and are only decrypted and processed in TEEs of blockchain nodes. 
Besides, unauthorized access and modification of transactions and the state of applications by organizations other than the owning and authorized nodes are disallowed. 

However, directly porting consortium blockchains into TEEs does not guarantee confidentiality, due to two persistent limitations inherent in TEEs. First, TEEs are vulnerable to micro-architectural side-channel leaks. Attackers can exploit state changes at either the micro-architectural level (e.g., CPU caches~\cite{brasser2017software,gotzfried2017cache,schwarz2017malware,moghimi2017cachezoom}) or architectural level (e.g., page faults~\cite{van2017telling, xu2015controlled} and DRAM contention~\cite{wang2017leaky}) to infer sensitive information. Second, TEEs lack state continuity guarantees. As TEEs store persistent states on untrusted storage, adversaries can replay outdated states, effectively rolling back execution~\cite{wang2022multi,matetic2017rote,niu2022narrator}.

These issues are inherently challenging to resolve, and to date, there are few widely deployed solutions to address them. Yet, they pose significant threats to the confidentiality of consortium blockchains~\cite{qi2024sok,li2024sok}. To understand their implications, we scrutinize four existing Confidential blockchains: Fabric Private Chaincode (FPC)~\cite{brandenburger2018blockchain}, Oasis~\cite{oasis2020}, CCF~\cite{howard2023ccf,russinovich2019ccf}, and CONFIDE~\cite{yan2020confidentiality}. We summarize and categorize two types of attacks: \eiatt attacks, which exploit side-channel vulnerabilities, and \eratt attacks, which leverage state rollback vulnerabilities in TEEs, to breach confidentiality.

\bheading{Our approach.}
Instead of addressing the two issues with general solutions, which introduce significant performance overhead and system complexity (see \secref{sec:related}), we propose a holistic approach with a novel architecture for Confidential blockchains, called \sysname.

\sysname introduces \textit{delegated execution} to defend against \eiatt attacks. 
The delegation scheme allows a client to appoint trusted nodes (called \textit{delegators}) to execute its transactions and synchronize the execution results with others, minimizing the \eiatt attack surface exposed to potentially malicious nodes. Other nodes can then trust the correctness of these results without re-execution, as TEEs ensure execution integrity. 
This scheme is inspired by the \textit{asymmetric trust}~\cite{Asymmetric, cachin2021asymmetric}, in which a client can declare a list of trusted nodes to order and execute its transaction for better security, privacy protection, as well as system efficiency. This design philosophy is widely used in blockchains such as Ripple~\cite{ripple}, Stellar~\cite{Stellar}, and Quorum Besu~\cite{besu}, however, it is first used in confidential blockchains. (See more details in \secref{sec: delegated execution}.) 
In addition, \sysname adopts the \textit{order-then-execute} architecture to address \eratt attacks. This architecture ensures that transactions are ordered before execution and prevents ordered transactions from being reverted, effectively mitigating rollback attacks at the blockchain level.

Despite the simplicity of these ideas, incorporating them into a confidential blockchain system is non-trivial. 
Delegated execution inevitably results in high delegation overhead, as transactions need to be sent to the delegators and synchronized to other nodes. Moreover, the delegation process is vulnerable to denial-of-delegation (DoD) attacks, where malicious delegators may refuse to execute transactions and block the system. To address these challenges, we propose a \textit{location-aware concurrent execution} mechanism that leverages the state dependency graph (SDG) to parallelize transaction execution, reducing the delegation overhead. We also design a \textit{delegation failure handler} to detect and recover from DoD attacks, ensuring the liveness of the system.

We developed a prototype implementation of \sysname with about 1K LoC C++ code atop the consortium blockchain platform, FISCO BCOS (referred to as BCOS for short)~\cite{FISCO_BCOS}.
We revised the logic of the underlying architecture including the consensus and execution layer. 
For the consensus layer, we constructed VR protocol~\cite{liskov2012viewstamped} based on the Raw PBFT protocol of BCOS. 
For the execution layer, we added the synchronization module to synchronize execution results from delegators to other nodes.
Considering that there are no competing implementations using confidential virtual machines, \eg, AMD SEV-SNP, Intel TDX~\cite{intel_tdx_website}, and that the performance of Intel SGX-based implementations~\cite{howard2023ccf, cheng2019ekiden, yan2020confidentiality} is not comparable due to varying hardware mechanisms, we did not compare \sysname with other existing systems.
We conducted extensive experiments to evaluate \sysname with AMD SEV-SNP~\cite{amd_sev_website} in both LAN and WAN conditions and compared it with BCOS. 
The results show that in a WAN environment with 3 nodes, \sysname can achieve about 4k throughput and keep the latency as low as 0.4-0.5s. 
Compared to the performance of BCOS, \sysname exhibits performance degradation introduced by SEV-SNP and the synchronization module.

\bheading{Our contributions.} Our contributions are as follows:
\begin{packeditemize}
\item  We identify two attacks, \ie, \eiatt and \eratt attacks, which leverage state rollback and side-channel vulnerabilities in TEEs, to breach confidentiality.
We scrutinize four existing Confidential blockchains. 

\item We propose \sysname, a novel Confidential blockchain incorporating delegated execution and order-then-execute architecture to mitigate \eiatt and \eratt attacks.

\item We present two mechanisms, \ie, location-aware execution, and delegation failure handler. The former parallelizes transaction execution to reduce the performance overhead of introduced delegation execution; the latter ensures system liveness in the presence of malicious delegators.  

\item We implement \sysname atop BCOS and evaluate its performance on a public cloud platform through extensive experiments. Evaluation results show that \sysname has a competitive performance.
\end{packeditemize}
\section{Background}
\subsection{Consortium Blockchain} \label{subsec:Consortium Blockchain} 
Consortium blockchains~\cite{androulaki2018hyperledger, quorum, li2023fisco, brown2016corda}, also known as permissioned blockchains, allow identity-known nodes to maintain an append-only ledger. Unlike public blockchains (\eg, Bitcoin), consortium blockchains restrict on-chain data access and validation to authorized nodes, whose participants typically include commercial enterprises or government entities. 
The permissioned features make them ideal for applications, such as cross-border exchange~\cite{islam2022low}, goods tracing~\cite{assaqty2020private}, and supply chain~\cite{wu2023high}, which prioritize privacy, scalability, and regulatory compliance. 
In such applications, a node usually processes transactions from its clients, who trust only that node. 
This \textit{asymmetric trust} inspires us to propose \textit{delegated execution}. 

Consortium blockchains consist of two core components. The first is Byzantine Fault Tolerant (BFT) consensus, which enables mutually distrusting nodes to agree on the same transaction sequence.   
Notable BFT consensus protocols include PBFT~\cite{Castro:2002:PBFT} and HotStuff~\cite{yin2019hotstuff}. 
By executing ordered transactions, nodes can reach the same state. Specifically, according to the sequence of ordering and execution, there are two architectures: \textit{order-then-execute} and \textit{execute-then-order}. 
BCOS, Corda, and Quorum adopt the former, whereas Hyperledger Fabric uses the latter. 
The second is smart contracts, which are programs to be executed on blockchains. Smart contracts are usually written in high-level languages like Solidity~\cite{Solidity} and compiled into bytecode and executed within virtual environments such as the Ethereum Virtual Machine (EVM)~\cite{evm_website}. 
Compared to simple payment transactions, smart contracts typically involve complex execution patterns, which can inadvertently leak information to an attacker to infer sensitive data.

\subsection{Trusted Execution Environment}
TEEs are CPU extensions that offer applications a secure execution environment to prevent running code or data from being accessed or tampered with by unauthorized entities. TEEs leverage techniques like hardware-assisted isolation, memory encryption, and remote attestation. 
Existing TEE platforms can be divided into two classes: process-based TEEs and Virtual Machine (VM)-based TEEs. 
Process-based TEEs like Intel SGX~\cite{intel_sgx_website} support process-based applications, whereas VM-based TEEs, such as AMD SEV~\cite{amd_sev_website} and Intel TDX~\cite{intel_tdx_website}, provide a virtualized environment running operating systems. Intel SGX is among the first generation of commercial TEE platforms, however, more and more manufacturers (\eg, Intel) are turning to VM-based TEEs. Thus, in this work, we implement our design atop VM-based TEEs, \ie, SEV-SNP, but, can also be extended to processor-based platforms.  

\section{Dissecting Existing Confidential Blockchains} \label{Consortium Blockchains}
We scrutinize the confidentiality of four confidential blockchains, including Fabric Private Chaincode (FPC)~\cite{brandenburger2018blockchain}, Oasis~\cite{oasis2020}, CCF~\cite{howard2023ccf, russinovich2019ccf} and CONFIDE~\cite{yan2020confidentiality}. Specifically, we propose two attacks, \eiatt attack and \eratt attack, which utilize side-channel attacks and state rollbacks of TEEs, respectively, to violate confidentiality.

\subsection{\EIatt Attack}
In an \eiatt attack, the adversary utilizes side-channel leakage, such as CPU cache utilization~\cite{gotzfried2017cache,moghimi2017cachezoom,brasser2017software,schwarz2017malware} and memory access patterns~\cite{van2017telling,xu2015controlled}, during the execution of the smart contracts to infer target sensitive information without decrypting the ciphertext directly. This attack does not influence the normal execution of the smart contracts~\cite{fei2021security,li2022systematic} but can be exploited to undermine fairness~\cite{desai2021secauctee}. 

\begin{figure}[t]
\centering
\includegraphics[scale=0.21]{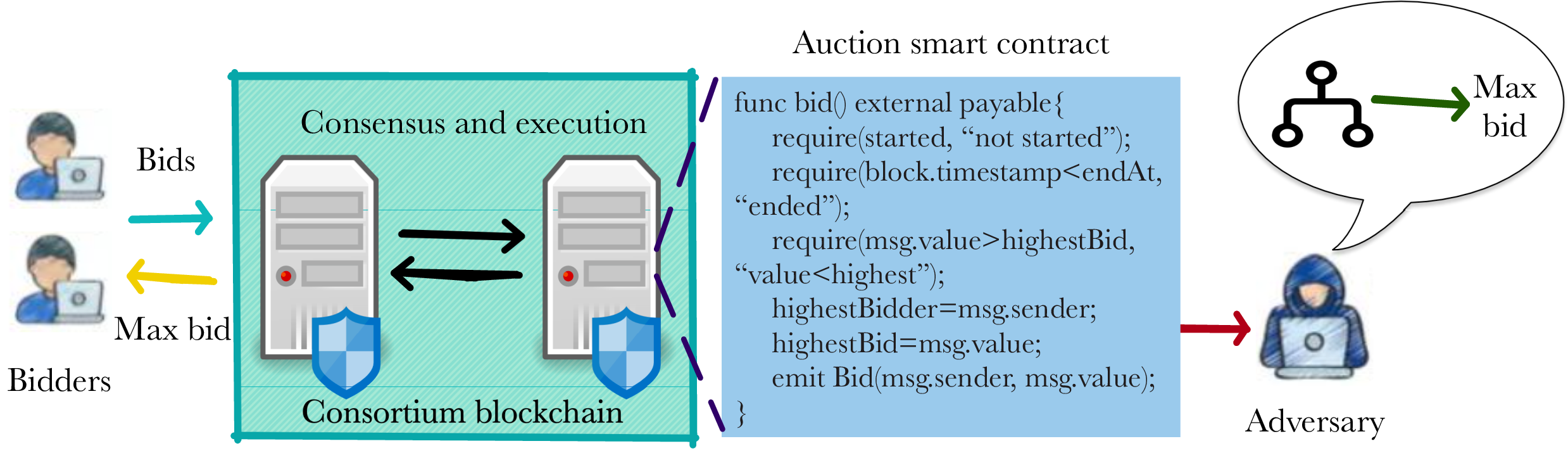}
\caption{An auction example running on a smart contract. The attacker can control one node, observe the running applications, and infer the secret to break fairness. }
\label{fig:toy}
\end{figure}

To illustrate \eiatt attacks, \figref{fig:toy} shows an on-chain auction implemented as a smart contract. 
Each bidder sends their bid in an encrypted transaction to the TEE nodes. 
Then, the transaction is decrypted inside the TEE and handled by the auction smart contract. 
The contract compares the received bid \textit{msg.value} in the transaction with the current highest bid \textit{highestBid} (Line 4-6), and updates the \textit{highestBid} and \textit{highestBidder} if the newly received bid is higher. 

During the above process, an adversary may control one or more TEE nodes, record the memory access patterns observed during the execution of the auction smart contract, and infer the control flow of the smart contract. For example, if the adversary can determine, by performing side-channel inference, whether the \textit{highestBid} and \textit{highestBidder} are updated during execution, she may be able to leverage such information to win the auction at lower cost. Thus, \eiatt attacks may lead to a breach of fairness.

\bheading{Existing platforms and their defenses.} 
We found that all existing Confidential blockchains, including FPC, Oasis, CCF, and CONFIDE, are susceptible to this attack.
In particular, FPC, CCF, and CONFIDE do not consider \eiatt attacks in their threat model and consequently have no defense. 
Oasis considers execution inference caused by side-channel attacks~\cite{chen2019sgxpectre, van2018foreshadow, wang2017leaky} and suggests isolating the SGX-based compute nodes from the shared environment to constrict access from malicious nodes. 
However, Oasis does not provide a detailed design, making it impossible to evaluate the effectiveness of the defense~\cite{oasis2020}. 

\subsection{\ERatt Attack} \label{subsec:rollback}
In an \eratt attack, an adversary crafts transactions, executes these transactions, and locally observes the outputs (\eg, winning bids from an auction), and then selects the preferred transaction (\eg, the lowest bid with priority to buy the item).
If execution order of transactions is not finalized by consensus, the adversary can repeat this procedure to execute different transactions one by one, analyze outputs and state updates by terminating and restarting the TEE running smart contracts, and eventually infer sensitive information to violate transaction confidentiality by exploiting rollback issues of TEEs~\cite{matetic2017rote, niu2022narrator, Achilles}. 
For instance, in the aforementioned auction case, the adversary can incrementally raise bids and compare them with the current highest bid \textit{highestBid} until finding a bid value marginally that surpasses the \textit{highestBid}.

The susceptibility of this attack is influenced by the underlying system architecture. As described in 
\secref{subsec:Consortium Blockchain}, existing blockchain architectures can be categorized into two types: \textit{execute-then-order} and \textit{order-then-execute}. 
In the former, nodes execute transactions, simulate state updates, and subsequently run consensus to commit the final state. The simulated state prior to consensus makes this architecture vulnerable to \eiatt attacks. 
Conversely, in the latter architecture, nodes first commit transactions through consensus and then deterministically execute the ordered transactions. The finality of committed transactions prevents the adversary from altering executed transactions, while deterministic execution ensures that repeated executions do not reveal differentiable outputs or state updates. As a result, the \eratt attack is effectively mitigated in this architecture.

\bheading{Existing platforms and their defenses.}
{CCF and CONFIDE use the order-then-execute architecture, making them resilient to the execution-replay attack. Perhaps for this reason, they do not consider \eratt attack in the design.
By contrast, FPC and Oasis adopt the execute-then-order architecture, making them vulnerable to \eratt attacks. 
To defend against these attacks, compute nodes in Oasis encrypt the outputs and state updates after processing transactions and reveal the decryption key to decrypt them only after monitoring that they are committed. 
However, this approach does not provide full protection because the adversary can leverage \eiatt attacks to infer secret information during execution without decrypting the ciphertext directly. 
The leaked side-channel information enables the adversary to craft transactions to launch \eratt attacks and violate the confidentiality guarantee of the system.
Furthermore, Oasis requires the compute node to be always online to reveal the decryption key, posing a strict requirement for availability and synchronous networks. FPC introduces human coercion to modify the smart contract code  (\eg, by adding a ``close" function) to establish barriers against such attack.
However, this method introduces significant human costs for smart contract development.

\bheading{Summary.} Both \textit{\eiatt} and \textit{\eratt attacks} can violate confidentiality of confidential blockchains. 
Specifically, the root cause of \eiatt attacks is the inherent, inevitable side-channel issues of TEEs.
All four platforms are susceptible to \eiatt attacks because none of them provides a complete architectural design to defend against side-channel issues.
Meanwhile, the root cause of \eratt attacks is the rollback issues of TEEs.
This attack can be launched in FPC and Oasis with an execute-then-order architecture. 
All of these motivate us to design a confidential blockchain that efficiently mitigates these attacks. 

\section{System Model and Problem Statement}
\subsection{System Model}
\label{sec:system models}
We consider a consortium blockchain system with a set of $n = 2f + 1$ nodes, denoted by $\mathcal{N} = \{p_1, p_2, ..., p_n\}$. Each node $p_i$ is equipped with a TEE-enabled machine.  
We also consider a set of clients, denoted by $\mathcal{C} = \{c_1, c_2, ...\}$.
Clients send their transactions to nodes for processing. 
We assume there exists a public-key infrastructure (PKI) among nodes and clients for distributing public keys. 
Specifically, each node $p_i$ has a key pair $(pk_i, sk_i)$ used
for signatures and encryption, where $pk_i$ is publicly available. 
For nodes with TEEs, the $sk$ is stored inside the TEEs.
A message $m$ signed with $sk_i$ is denoted by $\langle m \rangle_{\sigma_i}$. 
We assume a shared key pair $(pk_0, sk_0)$ within nodes' TEEs for transaction encryption. 
Specifically, clients use the public key $pk_0$ to encrypt transactions, while the TEEs use the private $sk_0$ to decrypt them. 

\bheading{Network model.} We assume normal nodes are fully and reliably connected: every pair of nodes is connected with an authenticated and reliable communication link.
We adopt the partial synchrony model of Dwork \etal \cite{dwork1988consensus}, widely used in many BFT consensus protocols~\cite{Castro:2002:PBFT, yin2019hotstuff, fast-hotstuff}.
In the model, there is a known bound $\Delta$ and an unknown Global Stabilization Time (\textsf{GST}), such that after \textsf{GST}, all message transmissions between two nodes arrive within a bound $\Delta$.

\subsection{Threat Model} \label{subsec:threat}
We assume that at most $f$ nodes are Byzantine, \ie, behaving in arbitrary ways, at any time, while the remaining ones are honest, \ie, strictly following the protocol.  
We consider the worst case, in which all Byzantine nodes can be controlled by a single adversary $\mathcal{A}$. Nevertheless, Byzantine nodes do not have full control of TEEs. Following prior studies, we make the following assumptions for TEEs: 

\begin{packeditemize}
    \item \textbf{Secure cryptographic code.} We assume the adversary cannot break cryptographic primitives (\eg, encryption or signatures), via cryptographic analysis or side-channel attacks. This implies that the cryptographic implementation inside the TEEs follows the state-of-the-art constant-time paradigm and is free of side-channel vulnerabilities. 
    
    \item \textbf{Secure TEE hardware.} TEEs are assumed to provide integrity and confidentiality guarantees in accordance with their hardware capabilities as prior works~\cite{fides, Achilles}. Specifically, TEEs support remote attestation, \ie, assuring the authenticity of TEE hardware and the initial states of TEE software. Moreover, TEEs adopt encrypted memory and CPU isolation, defending against attacks from both system software and system administrators. 
   
    \item \textbf{Control-flow side-channel leakage.}
    Micro-architectural side-channels through CPU cache~\cite{gotzfried2017cache,moghimi2017cachezoom,brasser2017software,schwarz2017malware}, page fault manipulation~\cite{xu2015controlled,shinde2016preventing,wang2017leaky,van2017telling,werner2019severest}, interrupt timing~\cite{van2018nemesis,li2021cipherleaks} are considered in scope. The adversary can observe memory access patterns of transactions and thus infer the control flow of their execution. However, as we assume secure cryptographic implementations, which is the practice case, cryptographic secrets cannot be leaked through side channels.
    
    \item \textbf{Rollback attacks.}
We assume that the state-continuity property of TEEs cannot be reliably preserved. This assumption stems from the fact that modern TEEs either lack monotonic counters~\cite{martin2021adam} or do not employ effective distributed counters~\cite{niu2022narrator} in practice. Consequently, an adversary can roll back the execution of smart contracts to outdated states and replay transactions multiple times.

\end{packeditemize}

\subsection{Problem Statement}
We follow prior works~\cite{COBRA} to model the Confidential blockchain as State Machine Replication (SMR), in which honest nodes work as a deterministic state machine and have the same state (the replicated service state) after executing the same sequence of transactions. 
In particular, nodes decrypt clients' encrypted transactions and execute them in TEEs to achieve confidentiality. 
Thus, we define a \textit{confidential SMR} service model, where nodes' TEEs agree on the same global state $S=\{tx_1, tx_2,..., tx_i\}$, composed of $i$  transactions ordered by their index. 
The confidential SMR service model aims to provide the following properties: 
\begin{packeditemize}
    \item \textbf{Safety.} 
    If two nodes $p_i$ and $p_j$ accept transactions $tx$ and $tx^{\prime}$ with the same index, respectively, then $tx = tx^{\prime}$. 

    \item \textbf{Liveness.} 
    An honest client's transaction $tx$ will eventually be committed and executed by honest nodes. 

    \item \textbf{Confidentiality.} 
    No adversary infers sensitive information from executing transactions.
\end{packeditemize}

The \textit{safety} and \textit{liveness} properties are the standard properties of SMR, while confidential SMR ensures a new property called \textit{confidentiality}. Specifically, to achieve the latter, our protocol has to defend against \eiatt and \eratt attacks under the threat model defined in \secref{subsec:threat}. 
Except for the three properties, our defense against the \eiatt and \eratt attacks should not impose significant overhead on performance, in terms of transaction throughput and latency.

\section{Our Approach}
In this section, we present two approaches: delegated execution and order-then-execute architecture. We then introduce the associated challenges. 

\subsection{Delegated Execution}
\label{sec: delegated execution}
We introduce \textit{delegated execution} to defend against \eiatt attacks. 
Delegated execution is inspired by the concept of \textit{asymmetric trust}~\cite{Asymmetric, cachin2021asymmetric} in distributed systems, in which clients select their trusted quorum of nodes for processing transactions.
For example, stellar consensus protocol (SCP)~\cite{lokhava2019fast, kim2019stellar} captures the \textit{subjective trust} and constructs federated Byzantine quorum systems in which nodes are divided into several quorum slices.
Each node needs to clarify its trusted quorum slices and accept any results gained from this slice.
In each slice, there is guaranteed to be at least one reliable node to process transactions and release the results of execution.
The processing results are synchronized to other quorum nodes trusted by other clients through available and fast network transmission channel.
Thus, the trusted results are ultimately broadcast to the global network and recorded in a ledger with consistent states as trust chain is conducted over the whole SCP network.
Furthermore, the \textit{asymmetric trust} model has been adopted in several permissionless blockchains like Ripple~\cite{ripple}, Stellar~\cite{Stellar} and consortium blockchains like Quorum Besu~\cite{besu}.
These systems justify the reasonableness of \textit{asymmetric trust}.
Moreover, in Confidential blockchains, nodes are not anonymous and have a specific trust basis with each other. 
Therefore, the \textit{delegated execution} scheme is feasible.

Except for the trust basis, delegated execution operates in Confidential blockchains based on the observation that execution integrity provided by TEEs enables nodes to synchronize execution results, thereby mitigating the impact of execution inference attacks from Byzantine nodes.
Specifically, by leveraging remote attestation and reliable communication links offered by TEEs, nodes can synchronize their state across the network without needing to re-execute each transaction on every node. This approach allows clients to appoint trusted nodes for transaction execution, thereby minimizing the \eiatt attack surface exposed to potentially malicious nodes.

\subsection{Order-then-Execute Architecture}
\label{sec: architecture analysis}
We adopt the \textit{order-then-execute} architecture to defend against \eratt attacks. As analyzed in \secref{subsec:rollback}, this architecture inherently mitigates such attacks by ensuring transactions are ordered before execution, thereby eliminating opportunities for adversaries to exploit rollback vulnerabilities during transaction processing.

In this architecture, nodes order transactions without immediate execution to verify validity. Thus, invalid transactions may be included in the blockchain. However, this does not pose a significant problem. During the execution phase, nodes can deterministically discard invalid transactions, ensuring the correctness of the blockchain state. This deterministic execution process is widely supported by most smart contract execution engines, such as EVM~\cite{evm_website} and WASM runtimes~\cite{eosio_blockchain, near_blockchain, EWasm}.
Moreover, consortium blockchains typically require clients to be authenticated.  Such authentication is effective in identifying and mitigating the submission of many invalid transactions. 

\subsection{Challenges and Solutions}
\label{sec: challenges}
Despite the simple idea behind the \textit{delegated execution}, its adoption introduces two key challenges, as outlined below.

\bheading{Challenge 1 \textbf{(C1)}: High delegation overhead.} The delegated execution incurs significant state synchronization overhead when different delegators execute transactions included in blocks. 
In blockchains, transactions are usually batched into blocks (\ie, data structure) to be processed together. This batching strategy can effectively amortize consensus overhead and improve throughput. However, when using the delegated execution scheme, transactions within a block may be assigned to different delegators for execution. As a result, these transactions must be executed sequentially across delegators, following their order. 
After processing, each delegator must synchronize its execution results with others. In a distributed network, such frequent synchronization introduces considerable overhead, severely impacting the system’s efficiency and rendering it impractical for many applications. 

\bheading{Solutions: Location-aware concurrent execution (\secref{sec: execution policy}).} 
This solution aims to improve the concurrency of executions to reduce the high overhead caused by serial state synchronization. 
To this end, we introduce the \textit{concurrency controller}, which adheres to two fundamental principles: i) transactions must be executed by the specified delegator, and ii) the correctness of execution must be guaranteed. 
Guided by the two principles, this solution aims to improve execution concurrency from two aspects: 
\begin{packeditemize}
    \item \textbf{Block level: } A block is divided into multiple sub-blocks, each assigned to one delegator (by Principle 1) for parallel execution. Besides, to ensure that the transaction execution across sub-blocks does not deviate from the original transaction order (by Principle 2),  the concurrency controller incorporates \textit{state synchronization mechanisms} to ensure the correctness and uniqueness of execution results.

    \item \textbf{Transaction level: } The concurrency controller uses the transactions' read/write operations to enhance execution parallelism. Specifically, transactions that access disjoint portions of a contract’s state can be executed concurrently without compromising Principle 2. 
\end{packeditemize}

\bheading{Challenge 2 \textbf{(C2)}: Denial-of-Delegation (DoD) attacks.} In this attack, a malicious client can delegate its transactions to a Byzantine delegator, who refuses to execute the transaction. As said previously, in SMR, nodes have to deterministically execute ordered transactions one by one to have the same state given the same sequence of transactions.  
Thus, the denial of executing the transaction will block all subsequent dependent transactions, leading to a liveness violation. 

\bheading{Solution: Delegation failure handler (\secref{sec: delegation failure handler}).} 
The failure handler contains two parts: failure detection and optimistic abort. 
First, nodes use a timeout mechanism to detect execution failure. Specifically, nodes set a timer after a block is ordered. 
When the timer is triggered, they broadcast a timeout message for execution. 
Second, when receiving $f+1$ timeout messages of execution, all nodes can abort execution of the delegated transaction and subsequent dependent transactions. 
The leader collects both successful execution results and aborted ones, and includes them in the following block to ensure all nodes agree on the same system's state.  
\section{Design Details}

\subsection{Overview}
We incorporate the above ideas into a novel Confidential blockchain architecture, termed \sysname. 
As shown in \figref{architecture}, nodes are equipped with TEEs to process transactions sent from users. Users interact with \sysname (\ie, submitting transactions and receiving replies) through the client interface.

\sysname processes transactions in two phases: ordering and execution. In the former, \sysname adopts a carefully designed TEE-based confidential consensus protocol in prior works~\cite{wang2022engraft, keyRecovery} (See more in \secref{subsec:implementation}), while in the latter, \sysname uses delegated execution. Specifically, to deal with the challenges above, \sysname uses location-aware concurrent execution (\secref{sec: execution policy}) and {delegation failure handler} (\secref{sec: delegation failure handler}). 

\begin{figure}[t]
\centering
\includegraphics[scale=0.2]{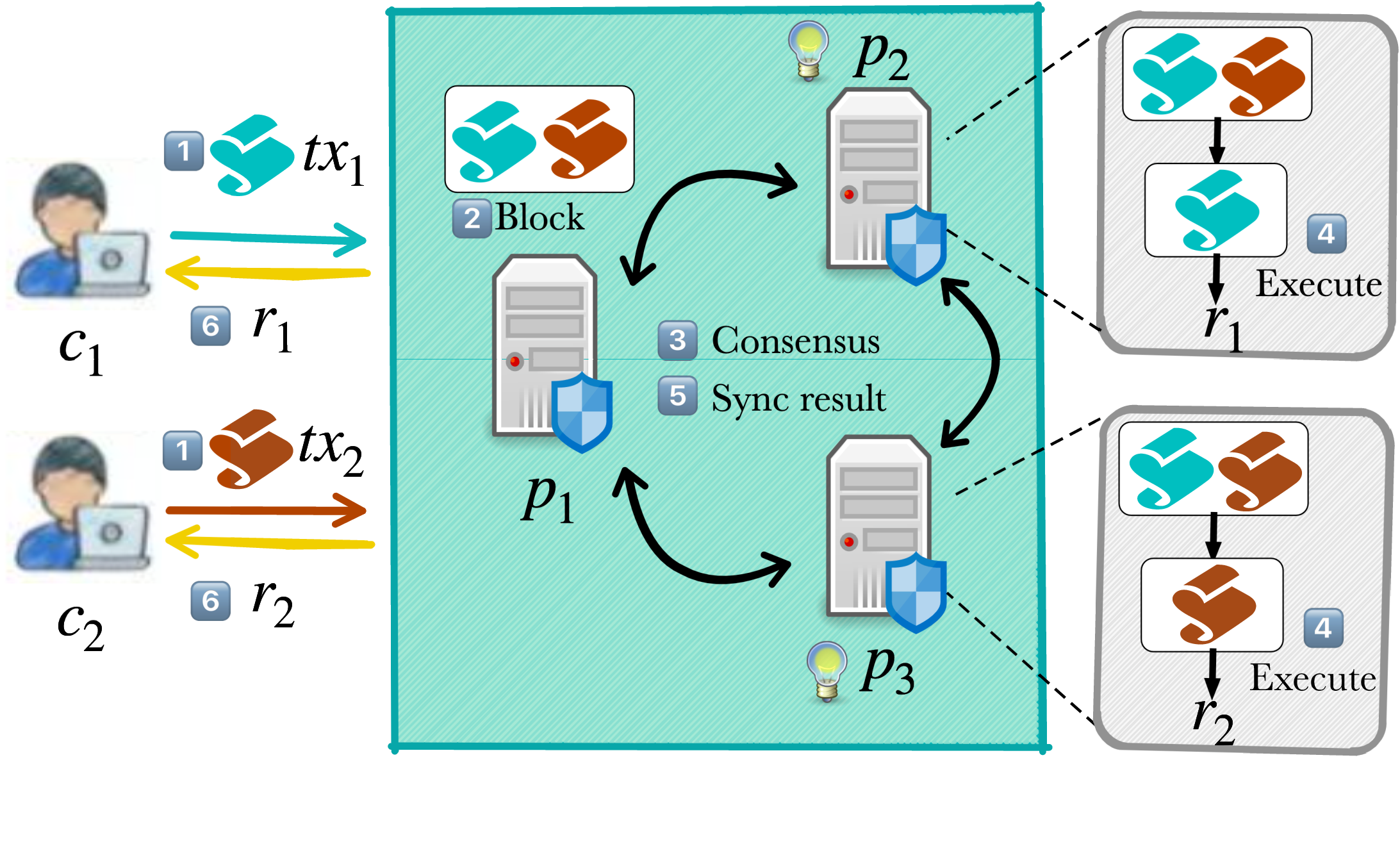}
\caption{A overview of \sysname architecture. For demonstration purposes, we set up two clients $c_1, c_2$ and three nodes $p_1, p_2, p_3$, all of which may become primary nodes or delegators during the consensus and execution phase.}
\label{architecture}
\vspace{-0.5cm}
\end{figure}

\bheading{Transaction flow.} \figref{architecture} shows an overview of the transaction workflow, in which three nodes $p_1$, $p_2$, and $p_3$ process two independent transactions $tx_1$ and $tx_2$ from clients $c_1$ and $c_2$, respectively. The nodes $p_2$ and $p_3$ are delegators for $tx_1$ and $tx_2$, respectively, and the node $p_1$ is the primary node in the current round.

\vspace{1mm} \noindent \blackding{1} Clients $c_1$ and $c_2$ send encrypted transactions $tx_1$ and $tx_2$ to nodes $p_1$, $p_2$, and $p_3$ for processing. 

\vspace{1mm} \noindent \blackding{2} The primary node $p_1$ batches transactions $tx_1$ and $tx_2$ into one block $B$ and then coordinates with nodes $p_2$ and $p_3$ to commit it to the blockchain by running the {confidential consensus protocol}.

\vspace{1mm} \noindent \blackding{3} $p_1$ generates the state dependency graph (SDG) and sends it to other nodes.
All nodes execute transactions based on the SDG to determine the execution sequence.
In this case, delegators $p_2$ and $p_3$ extract $tx_1$ and $tx_2$ from the block and execute them to obtain execution results $R_1$ and $R_2$, respectively.

\vspace{1mm} \noindent \blackding{4} The execution results $R_1$ and $R_2$ are synchronized to other nodes. If a malicious delegator, \eg, $p_2$, refuses to execute or send out the execution results, the delegation failure handler will abort the execution of transaction $tx_1$. All nodes will only accept execution result $R2$ and update their state. 

\vspace{1mm} \noindent \blackding{5} After synchronizing execution results, nodes return the execution results to the clients who sent the corresponding transactions.

\subsection{Location-aware Concurrent Execution}
\label{sec: execution policy}

\subsubsection{Trust Isolation Mechanism}
Transactions (including smart contracts) should be executed at nodes in the corresponding trust domain $\mathcal{D}$, referred to as the trust isolation mechanism. 
Specifically, a client specifies a $\mathcal{D}$ that includes a group of nodes while deploying a smart contract.
This action announces that these nodes are trusted by the initiator of the smart contract, and each transaction calling this smart contract has to be processed among them.
Thus, when initiating transactions to invoke this smart contract, clients must select a node called delegator $d_i$ within $\mathcal{D}$ as the transaction execution node. 
Specifically, when there is a call to an external contract in the smart contract, \ie, a cross-contract call, which is usually a read of the state of smart contracts.
The external call will be executed at $d_j$ in another trust domain specified by the external contract.

\subsubsection{State Dependency Graph}
The state dependency graph (SDG) is utilized for the concurrent execution of transactions.
We define the transaction as a set of to-be-read states $R(tx)$ and to-be-written states $W(tx)$. 
The transactions packaged into a block have a fixed order.
In a linear execution pattern, the order of transactions represents the sequence of their execution.
If transaction $tx_i$ precedes transaction $tx_j$ in the order, it is denoted as $T(tx_i)<T(tx_j)$.
We define the state dependency as follows:

\begin{definition}
    Given two transactions $tx_i$ and $tx_j$ with $T(tx_i)<T(tx_j)$, there is a state dependence between the two if there is an intersection between $W()$ of either party and $W()$ or $R()$ of the other party, recorded as $tx_i \rightarrow{tx_j}$.
\end{definition}

When multiple transactions are executed concurrently, state dependency between transactions can lead to conflict and uncertainty of execution results.
The SDG can be used to represent the state dependence between transactions.
The definition of it is given below.

\begin{definition}
The state dependency graph is a directed acyclic graph where nodes represent transactions within a block and directed edges represent state dependencies between transactions. It can be formalized as $G=(N, V)$, where $N=\{ tx_1, tx_2,..., tx_n\}$ and $V=\{(tx_i, tx_j)|tx_i \rightarrow{tx_j}\}$.
\end{definition}

We assume that all read/write operations of a smart contract are known, while the concurrency controller can obtain the specific read/write operations contained in each transaction without executing the transaction.
We use the SDG generation algorithm in~\cite{yao2015dgcc} to generate a raw SDG. 
The SDG is generated by the node proposing the new block during the consensus phase (\ie, the primary node) and synchronized to other nodes after the block consensus is completed.
In conjunction with location information related to delegators, the SDG is further decomposed to generate numerous subgraphs utilized by \textit{concurrency controller} module of delegators.

\begin{table}[t]
\centering
\caption{Information about smart contracts and transactions in the example.}
\label{tab:example info}
\begin{adjustbox}{width=1\columnwidth}
\begin{tabular}{c|c|c|c|c}
\toprule
\textbf{Smart contract} & \textbf{Operational state} & \textbf{Transaction} & \textbf{Delegator} & \textbf{Sub-block} \\
\midrule
$S_1$ & $a, E(S_2:b)$ & $tx_1, tx_4$   & $d_1: \mathcal{D}_1$ & $SB_1$  \\
$S_2$  & $b, E(S_3:c)$ & $tx_2, tx_5, tx_6$   & $d_2: \mathcal{D}_2$ & $SB_2$            \\
$S_3$ & $c$ & $tx_3$ & $d_3: \mathcal{D}_3$ & $SB_3$ \\
\bottomrule
\end{tabular}
\end{adjustbox}
\end{table}

\subsubsection{Concurrency Controller}
The concurrency controller executes concurrency at both block and transaction levels based on the SDG, as introduced below. 

\bheading{Concurrency of sub-blocks.} 
When a block is committed, the included transactions are divided into several sub-blocks, each assigned to one delegator. 
Based on these sub-blocks and dependencies across transactions, an SDG is generated for each sub-block, guiding the execution sequence of each transaction.

Sub-blocks can be formally defined as follows: assuming the existence of blocks $B=\{tx_1, tx_2,..., tx_{n-1}, tx_n\}$ and each transaction corresponds to a delegator, then $B$ can be further expressed as $B'=\{tx_1(n_i), tx_2(n_j),..., tx_{n-1}(n_e), tx_n(n_f)\}$. $tx_1(n_i)$ and $tx_2(n_j)$ belong to the same sub-block when and only when $n_i=n_j$, which can be expressed as $SB_i=\{tx_1, tx_ 2\}$.

\iheading{Example}: Table \ref{tab:example info} lists three abstracted smart contracts $S_1, S_2, S_3$, and the corresponding operational states where $E()$ means that there is a cross-contract call in the contract, and in parentheses are the target contract and the contract state to be called. 
We assume that the client sends the following six transactions $\{tx_1, tx_2, tx_3, tx_4, tx_5, tx_6\}$, and assigns a delegator to each transaction.
The six transactions are packaged into one block. 
Eventually, the concurrency controller will divide the block into three sub-blocks $SB_1, SB_2, SB_3$ based on the delegator.

\bheading{Concurrency of read/write operations.}
The essence of a transaction is to read or write to the state of the smart contract. 
From the respective read/write operation sequences of transactions, it can be observed that the operations of $tx_1$ and $tx_3$ on resources do not conflict with each other. 
Therefore, they can be executed in parallel. 
However, the execution of $tx_2$ and $tx_5$ depends on the results of the two transactions above.
So they need to wait until those transactions are executed before proceeding.
\figref{SDG} shows the corresponding read/write set for each transaction and the generated subgraphs of SDG.

All transactions are executed in the order prescribed by SDG.
For example, among the above six transactions, $tx_1$ and $tx_4$ belong to the sub-block $SB_1$, so the execution location of the transaction is $d_1$, and similarly, $tx_2$, $tx_5$, $tx_6$ are executed at $d_2$. $tx_3$ is executed at $d_3$.
When there is a state dependency on the transactions executed at different delegators, the execution of the transactions will be halted until the results of the preceding sequence-dependent transactions are received.

The execution sequence at $d_1$ is: execute $tx_1$ $\to$ detect $tx_2$ has dependency on $tx_1$, send the execution result to $d_2$ $\to$ detect $tx_4$ has dependency on the execution result of $tx_2$, and then start to wait for the execution result $\to$ receive the $StateSyncMsg$ from $d_2$ and update the local state then start executing $tx_4$ $\to$ send the current state of the smart contract to $d_2, d_3$ when execution is complete.
The handling process at other nodes is similar.

\begin{figure}[t]
\centering
\includegraphics[scale=0.12]{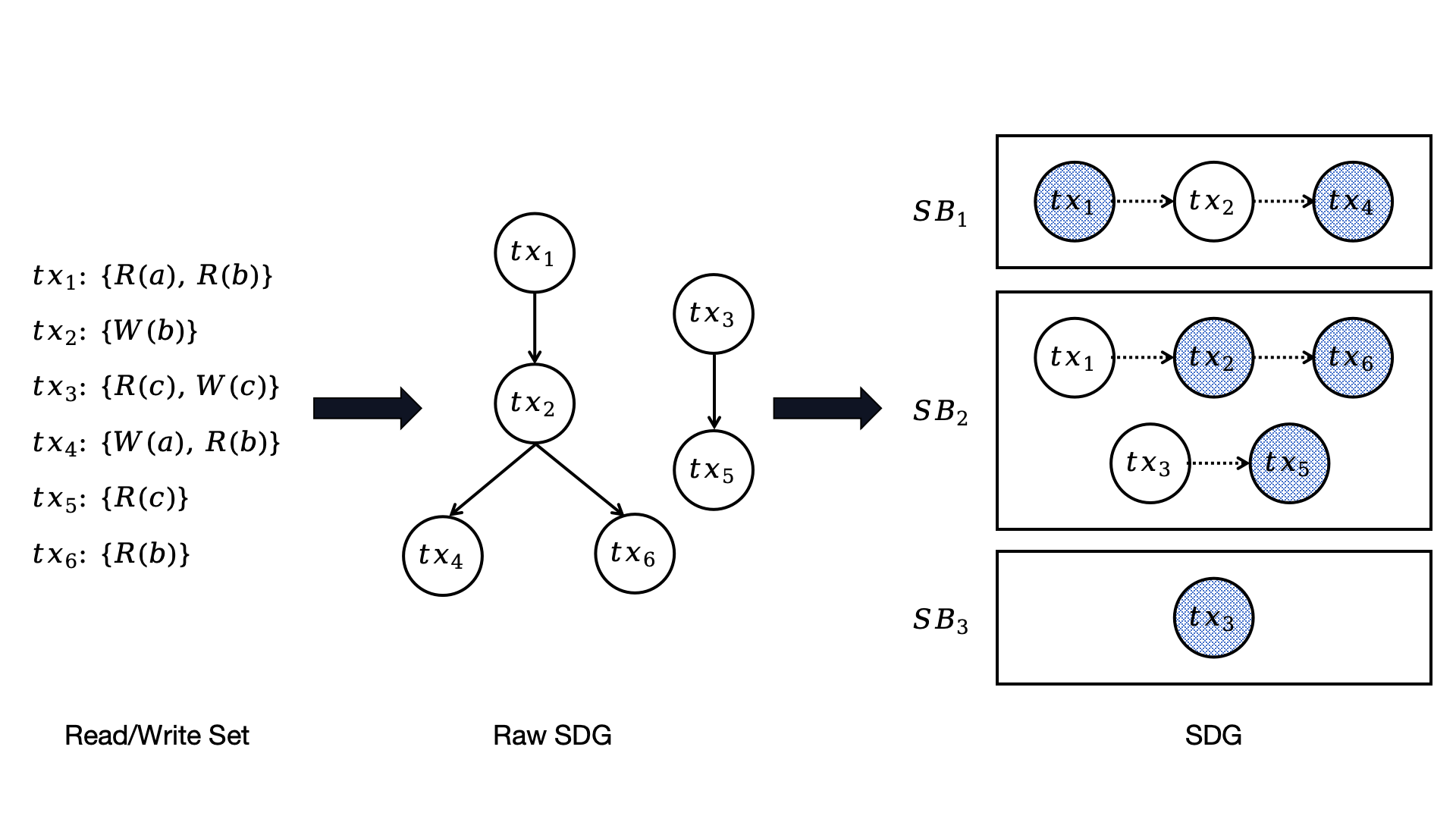}
\caption{SDG of the example. Given a series of transactions $tx_1,..., tx_6$ with fixed order, the primary node analyzes the read/write set of transactions and determines the conflict situation of execution to generate raw SDG. According to the delegators specified by transactions, the raw SDG is split into multiple sub-SDGs to provide parallel support for them.}
\label{SDG}
\end{figure}

\subsubsection{State Synchronization}
We introduce the state synchronization mechanism to synchronize the local execution results to other nodes promptly and ensure that transactions are always executed based on the correct smart contract state.
Overall, state synchronization includes two types of scenarios:

\bheading{Partial synchronization.} 
Suppose a node $p_1$ has not finished executing all the transactions in a sub-block $SB_i$, but there is a transaction $tx_j$ in another sub-block $SB_j$ to be executed in node $p_2$ that has state dependency on the transaction $tx_i$ in $SB_i$. 
In that case, it is necessary to instantly send the execution results of $tx_i$ to $p_2$.

\bheading{Full synchronization.} After a node executes all transactions within a sub-block, it sends the latest local state of smart contracts to all other nodes for synchronization.

After receiving it, the node processes the \textit{StateSyncMsg} to update the smart contract state locally by Algorithm \ref{algorithm}.
The structure of \textit{StateSyncMsg} message is $\langle \mathcal{I}_B, \mathcal{I}_{SB}, r, S, v, h \rangle_{\sigma_1}$, where $\mathcal{I}_B$ is the sequence number of the block.
$\mathcal{I}_{SB}$ is the sequence number of the sub-block divided by $B$.
$\mathcal{I}_B$ and $\mathcal{I}_{SB}$ are used to identify the sub-block.
$r$ represents whether the results stored in the message are for all transactions.
$S$ is the execution result of transactions in $\mathcal{I}_{SB}$.
$v$ is the version of smart contracts.
$h$ is the hash value of $S$ and is used to guarantee $S$ is not revised.

The node determines whether the smart contract is the latest version $v$ before executing to prevent transactions from being executed based on the wrong smart contract state.
Specifically, each transaction is executed based on a specific version of the smart contract.
All nodes maintain a public and single version $v$ based on the current local version of the smart contract.
When a transaction exists that writes a new state to the smart contract, we add one to $v$ of this smart contract.

\begin{algorithm}[t]
\caption{State Update Protocol}
\label{algorithm}
\begin{algorithmic}[1]
    \STATE \textbf{upon} $msg\_queue$ is not null
    \STATE \hspace{1.0em} $handler \gets thread.initial()$
    \STATE \hspace{1.0em} $msg \gets node.msg\_queue.push()$
    \STATE \hspace{1.0em} \textbf{if} $msg == StateSyncMsg$ \textbf{then}
    \STATE \hspace{2.0em} $state \gets handler.get\_state(msg)$
    \STATE \hspace{2.0em} $r \gets handler.get\_r(msg)$
    \STATE \hspace{2.0em} $v \gets handler.get\_version(msg)$
    \STATE \hspace{2.0em} \textbf{if} $r==false$ \textbf{then} 
    \STATE \hspace{3.0em} \textbf{if} $v==v'+1$ \textbf{then}
    \STATE \hspace{4.0em} $handler.overwrite(state)$
    \STATE \hspace{4.0em} $v'=v$
    \STATE \hspace{3.0em} \textbf{else}
    \STATE \hspace{4.0em} $node.msg\_queue.put(msg)$
    \STATE \hspace{2.0em} \textbf{else}
    \STATE \hspace{3.0em} \textbf{if} $v>v'$ \textbf{then}
    \STATE \hspace{4.0em} $handler.overwrite(state)$
    \STATE \hspace{4.0em} $v'=v$
    \STATE \hspace{3.0em} \textbf{else}
    \STATE \hspace{4.0em} discard $msg$
\end{algorithmic}
\end{algorithm}

\subsection{Delegation Failure Handler}
\label{sec: delegation failure handler}
This section introduces the delegation failure handler to ensure that Byzantine delegators cannot disrupt transaction execution and compromise system liveness. 
Specifically, suppose a Byzantine delegator does not promptly execute the transactions in sub-blocks. 
In that case, the subsequent transactions that should be executed based on the execution results of these transactions will be blocked.
\sysname can detect failed transactions and achieve consensus on the execution status of all transactions within the block globally. 
Subsequently, the primary node continues processing the transactions to ensure the liveness of \sysname.
The handler contains two components: failure detection and optimistic abort.

\bheading{Failure detection.}
When a new block $B$ completes consensus, the primary node notifies other nodes of the consensus result, \ie, $B$ has been committed to the blockchain. 
These nodes start a timer upon receiving this notification, and the delegators begin executing the transactions in the sub-blocks. 
The timeout mechanism introduces a time threshold $\theta$. Each node checks whether its local smart contract state is up-to-date if the accumulated computation time $t$ exceeds $\theta$. 

Specifically, as transactions in the block are executed, the version numbers of the corresponding smart contracts should monotonically increase. From the transactions in the block, the node can compute the expected version number $v'$ of the smart contracts after all transactions are successfully executed. Assuming the local version number of the smart contracts is $v$ when the timeout mechanism is triggered, the node compares $v$ with $v'$. If $v \neq v'$, this indicates that some sub-blocks failed to execute. 
The node then sends a \textit{timeout} message to the primary node. 
This message includes the index of transactions where the failure occurred.

\bheading{Optimistic abort.}
Upon receiving $f+1$ \textit{timeout} messages, the primary node begins collecting the execution results of the transactions within the block.
First, the primary node collects local transaction execution results, which originate from two sources: local execution and synchronization with other nodes. 
Next, it extracts transaction execution information from the timeout messages received from other nodes. 
Based on this information, the primary node compiles a list of all successfully executed transaction results and the failed transaction index. 
This summary is then packaged into a new block as a transaction. 
During the new consensus phase, the primary node broadcasts this new block to other nodes for consensus. 
Upon receiving the block, other nodes verify the transaction execution status and check if their state aligns with the reported status. If a node finds its state lagging, it requests the missing transaction execution results from other nodes to synchronize its state.

The above mechanism effectively prevents Byzantine delegators from compromising the overall system's liveness by refusing to execute transactions, as it discards transactions that remain unexecuted for an extended period.
\section{Correctness Analysis}
\subsection{Safety Analysis}
\sysname mainly includes two parts: consensus layer and execution layer, from which we analyze the testimony on which \sysname guarantees the safety attribute.
Since \sysname leverages a carefully designed confidential consensus protocol to provide ordering service, we focus on the safety of the execution layer in this section. For the consensus layer, we have Lemma~\ref{lemma:ordering}.

\begin{lemma}[Consensus Safety~\cite{wang2022engraft, keyRecovery}.]\label{lemma:ordering}
    The TEE-based Byzantine fault tolerance protocol generates a globally consistent order for received transactions before committing to the execution layer.
\end{lemma}

The nodes of the VR protocol running in the TEEs do not need to interact with the external hard disk, \ie, read or write data, so that the VR protocol can keep state continuity during transaction consensus.
The adversary cannot replay a transaction that has completed consensus.
We have Lemma~\ref{lemma:state consistency}.
Since \sysname follows the order-then-execute architecture, the execution layer receives transactions in the exact order. 

\begin{lemma}[Execution Safety.]\label{lemma:state consistency}
    The state of all nodes is ultimately consistent whether or not the delegator successfully executes the transaction.
\end{lemma}

\textsc{Proof.} In \sysname, transactions that complete the consensus may be executed at different nodes. 
Still, the execution order of transactions with dependencies does not change, and the sequence of transactions without dependencies does not affect state consistency.
When Byzantine nodes disrupt the normal execution of transactions, the delegation failure handler collects the execution status of all transactions and synchronizes it across all nodes. 
This ensures that the states of all nodes in the system remain consistent.
To this end, we can get Theorem~\ref{thm:safety}. 

\begin{theorem}[Safety.] \label{thm:safety}
    There are no two different $S$ in \sysname, nor are the versions $v$ of smart contracts.
\end{theorem}

\textsc{Proof.} The consensus layer ensures the deterministic order of transactions, \ie, only one $S$ exists in \sysname.
For the execution layer, location-aware concurrent execution leverages SDG to handle the conflict during concurrent execution.
The state consistency of smart contracts can be guaranteed.

\subsection{Liveness Analysis}
The guarantee of liveness requires the joint support of the consensus layer and execution layer.
Keeping the liveness of the consensus layer is a prerequisite for the liveness of the execution layer.
We leverage the confidential consensus protocol to ensure the liveness of the consensus layer.
So we have Lemma~\ref{lemma:liveness}.

\begin{lemma}[Consensus Liveness~\cite{wang2022engraft, keyRecovery}.]\label{lemma:liveness}
    The consensus layer can finish consensus for received blocks when no more than $f$ transactions are faulty simultaneously.
\end{lemma}

When the consensus layer maintains the liveness property, the execution layer receives transactions in the same order.  
In the execution layer, \sysname uses the delegated execution mechanism in which we have Observation~\ref{ob:delegated-exec}.

\begin{ob}\label{ob:delegated-exec}
    For a block generated by the consensus layer, each node may only be responsible for executing a sub-block and not have access to execute all transactions in the block.
\end{ob}

According to Observation~\ref{ob:delegated-exec}, there are two kinds of states in the execution layer:

\begin{packeditemize}
    \item $s_1$: Transaction execution at node $p_1$ can be halted due to depending on the result of transaction execution at node $p_2$.

    \item $s_2$: Inconsistent state at each node after all transactions in the block have been executed.
\end{packeditemize}

In $s_1$, the state synchronization mechanism is triggered after $p_2$ executes the transaction.
Then $p_2$ synchronizes the execution result to $p_1$ through the network to ensure the normal execution of the transaction at $p_1$. 
In the state $s_2$, each node triggers the state synchronization mechanism after executing all transactions in the sub-block.
The network environment and other reasons may lead to confusion in the order of synchronization state.
To avoid the latest state being covered by the old state, we leverage the version of smart contracts to ensure the correctness of the overwritten state.

When nodes are crashed or aborted maliciously, we have Observation~\ref{ob:dos}.
\begin{ob}\label{ob:dos}
    Assuming that transaction $tx$ is specified to be executed at node $p_1$ but $p_1$ suffers from an unexpected failure and can not restart in a short time, it can not execute $tx$ successfully.
\end{ob}

Nodes trigger the delegation failure handler mechanism to handle the aborted transactions.
So, transactions that fail in execution do not keep blocking the execution of subsequent transactions.
The failure handler mechanism will abandon these transactions and commit failed state of them to consensus layer.
So we have Lemma~\ref{lemma:executionliveness}.

\begin{lemma}[Execution Liveness.]\label{lemma:executionliveness}
    \sysname can process transactions and make progress despite failures or adversarial interference.
\end{lemma}

\textsc{Proof.} The proof is straightforward.
Nodes send a \textit{timeout} message to the primary node to trigger its optimistic abort mechanism. 
The primary node collects the transaction execution statuses and submits them in a new block for consensus. The execution results of completed transactions and information about aborted transactions are synchronized across all nodes. 
Transactions in the new block will then continue to be executed by delegators.

\begin{theorem}[Liveness.] \label{thm:liveness}
    Any valid transaction submitted by an honest client is eventually included in the blockchain and becomes immutable (finalized), provided the delegator operates under the assumed fault model.
\end{theorem}

\textsc{Proof.} With Lemma~\ref{lemma:liveness} and Lemma~\ref{lemma:executionliveness}, we have the liveness guarantee of \sysname.

\subsection{Confidentiality Analysis}
We first present that \sysname can prevent \ERatt and \EIatt attacks and then prove the confidentiality of \sysname.

\begin{lemma}[\ERatt Prevention.]\label{lemma:execution-reply prevention}
    Each executed transaction cannot be rolled back.
\end{lemma}
\textsc{Proof.} \sysname adopts the order-then-execute architecture, where nodes will only execute ordered transactions. The correctness of the consensus layer ensures that an ordered transaction will not be rolled back. Thus, this lemma is proven.

\begin{lemma}[\EIatt Prevention.]\label{lemma:execution-inference prevention}
    The execution of a transaction cannot be observed by nodes that are not delegated by the client.
\end{lemma}
\textsc{Proof.} With delegated execution, only delegators will execute transactions, and other nodes only see the execution results. As such, this Lemma holds.

\begin{theorem}[Confidentiality.] \label{thm:confidentiality}
    An adversary cannot obtain sensitive information through \attack attacks.
\end{theorem}
\textsc{Proof.} This theorem holds as a result of Lemma~\ref{lemma:execution-reply prevention} and Lemma~\ref{lemma:execution-inference prevention}.

\section{Evaluation}
We evaluate the performance of \sysname, aiming to answer the following questions:

\begin{packeditemize}
    \item \textbf{Q1:} How does \sysname perform compared to raw BCOS? (\secref{subsec:evaluation})

    \item \textbf{Q2:} How much performance overhead is introduced by SEV-SNP and other security modules of \sysname compared to BCOS? (\secref{subsec:overhead})
\end{packeditemize}

\subsection{Implementation and Setup}
\subsubsection{System Implementation} \label{subsec:implementation}
We build a prototype of \sysname atop FISCO BCOS v3.0.1 implemented in C++. 
We use AMD SEV-SNP as the TEE platform. 
We use bcos-crypto library\footnote{https://github.com/FISCO-BCOS/bcos-crypto} to provide message encryption/decryption and signature mechanisms.
We use bcos-boostssl\footnote{https://github.com/FISCO-BCOS/bcos-boostssl} to build secure and encrypted communication channels between nodes.
Clients interact with \sysname through RPC to send encrypted transactions or receive replies. 
We emulate the process in which the delegator synchronizes the execution state to other nodes using the \texttt{stateSync} functions. 
Specifically, we generate and fill \texttt{stateSync message} with randomly generated data so that its size matches the size of the execution results.
Therefore, this simplification does not affect the evaluation results of performance.

FISCO BCOS uses PBFT~\cite{Castro:2002:PBFT} as its consensus protocol. However, PBFT has $O(n^2)$ message complexity to commit transactions and requires $n = 3f+1$ nodes to tolerate $f$ Byzantine nodes. Prior works~\cite{wang2022engraft, keyRecovery} show that porting CFT protocols within TEEs with tailored design can achieve Byzantine fault tolerance.
The key idea is to use in-memory CFT protocols, which run consensus protocols in memory and do not use external storage to store state information. This can effectively defend against rollback issues of TEEs. 
To this end, we implement VR~\cite{liskov2012viewstamped} on FISCO BCOS. 
The extensions include reducing the consensus phases, message processing logic, the required number of votes, and the view-change mechanism. 
To simplify the system setup, we consider all nodes belonging to the same trust domain and choose the primary node as the delegator for its proposed block. 
After consensus, the block proposed is directly executed by the primary node. 
Then, it synchronizes the execution results of the block to other nodes directly.
This setup aligns with our design architecture and improves system performance compared to other setups.

\bheading{Baselines.}
We compare \sysname to FISCO BCOS to answer Q1.
We also consider three variants of \sysname to quantify the impact of the consensus, execution modules, and SEV-SNP on performance. 
The first variant, called \sysname-CS, runs VR protocol within SEV-SNP without the delegated execution. 
The second one, called \sysname-ES, allows nodes to be equipped with the delegated execution module and SEV-SNP, but not the VR protocol, called \sysname-ES.
The third one is BCOS-S, which directly runs BCOS within SEV-SNP to test the performance overhead of SEV-SNP.
The consensus protocol running in BCOS is PBFT.
By comparing \sysname with these variants, we aim to answer Q2.

\begin{figure*}[htbp]
\centering
\subfigure[\sysname]
{
    \begin{minipage}[b]{.23\linewidth}
        \centering
        \includegraphics[scale=0.55]{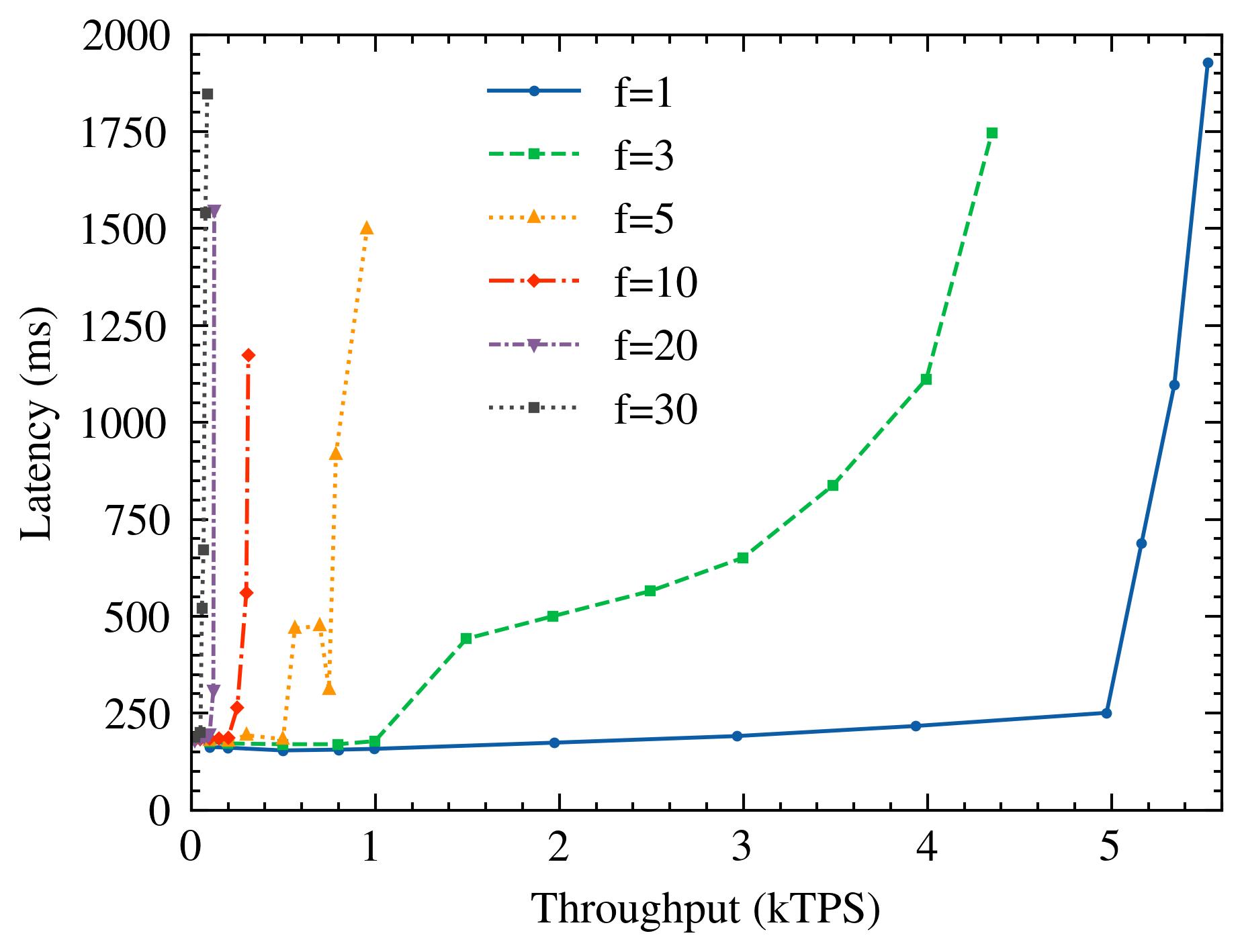}
        \label{CB_lan}
    \end{minipage}
}
\subfigure[FISCO BCOS]
{
 	\begin{minipage}[b]{.23\linewidth}
        \centering
        \includegraphics[scale=0.55]{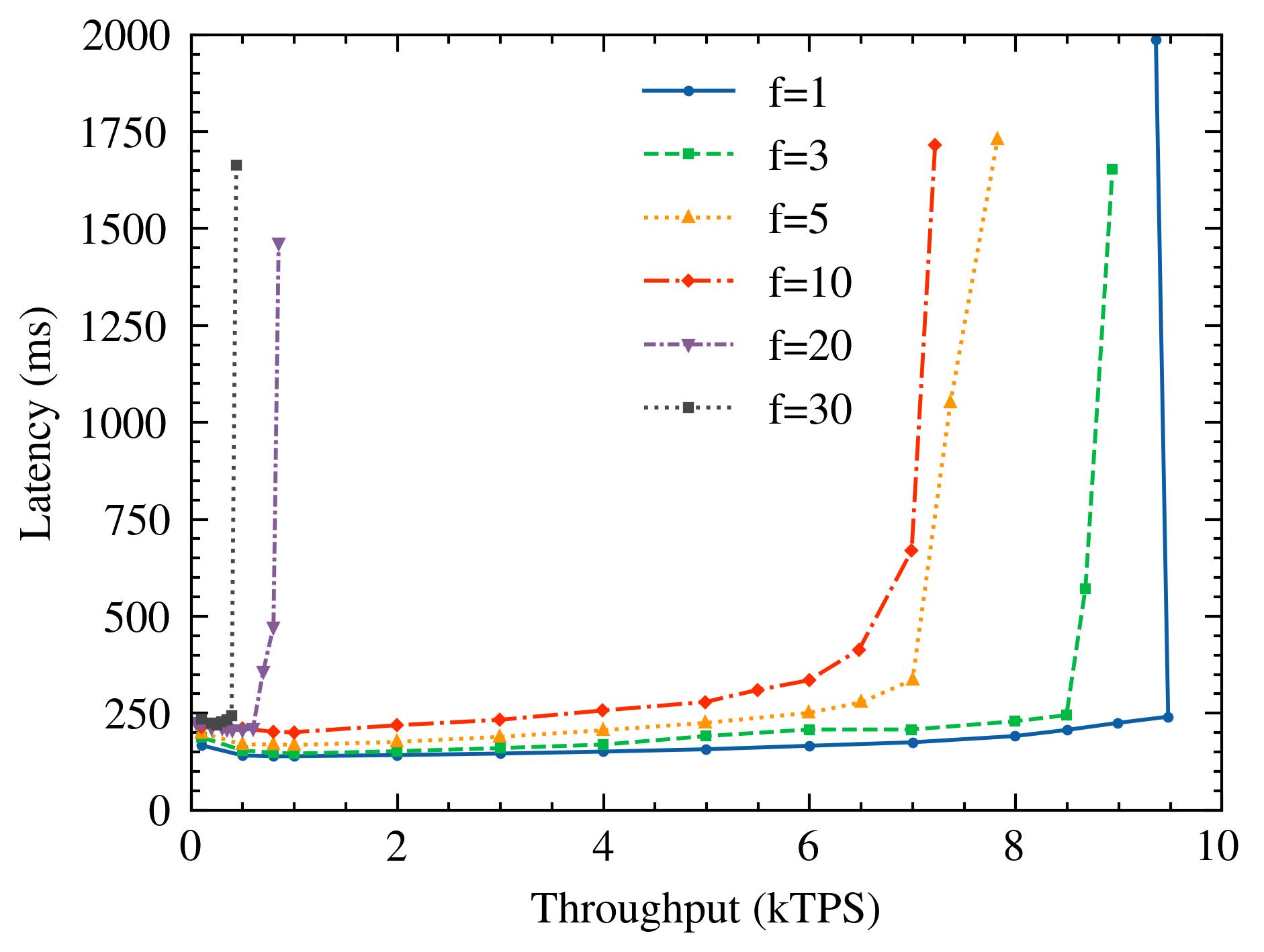}
        \label{bcos_lan}
    \end{minipage}
}
\subfigure[Maximum TPS vs. faults]
{
 	\begin{minipage}[b]{.23\linewidth}
        \centering
        \includegraphics[scale=0.55]{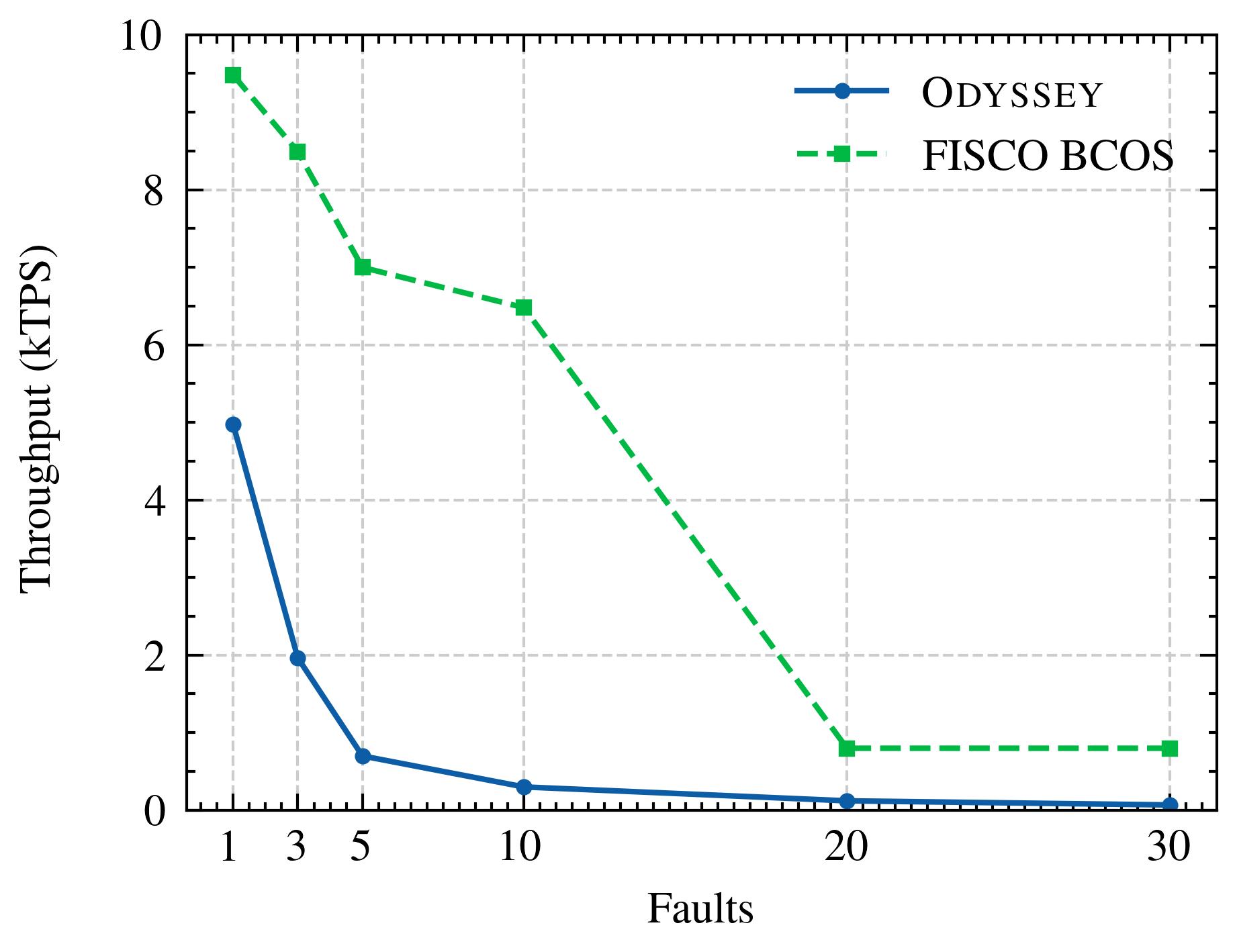}
        \label{lan_max_tps}
    \end{minipage}
}
\subfigure[Maximum latency vs. faults]
{
 	\begin{minipage}[b]{.23\linewidth}
        \centering
        \includegraphics[scale=0.55]{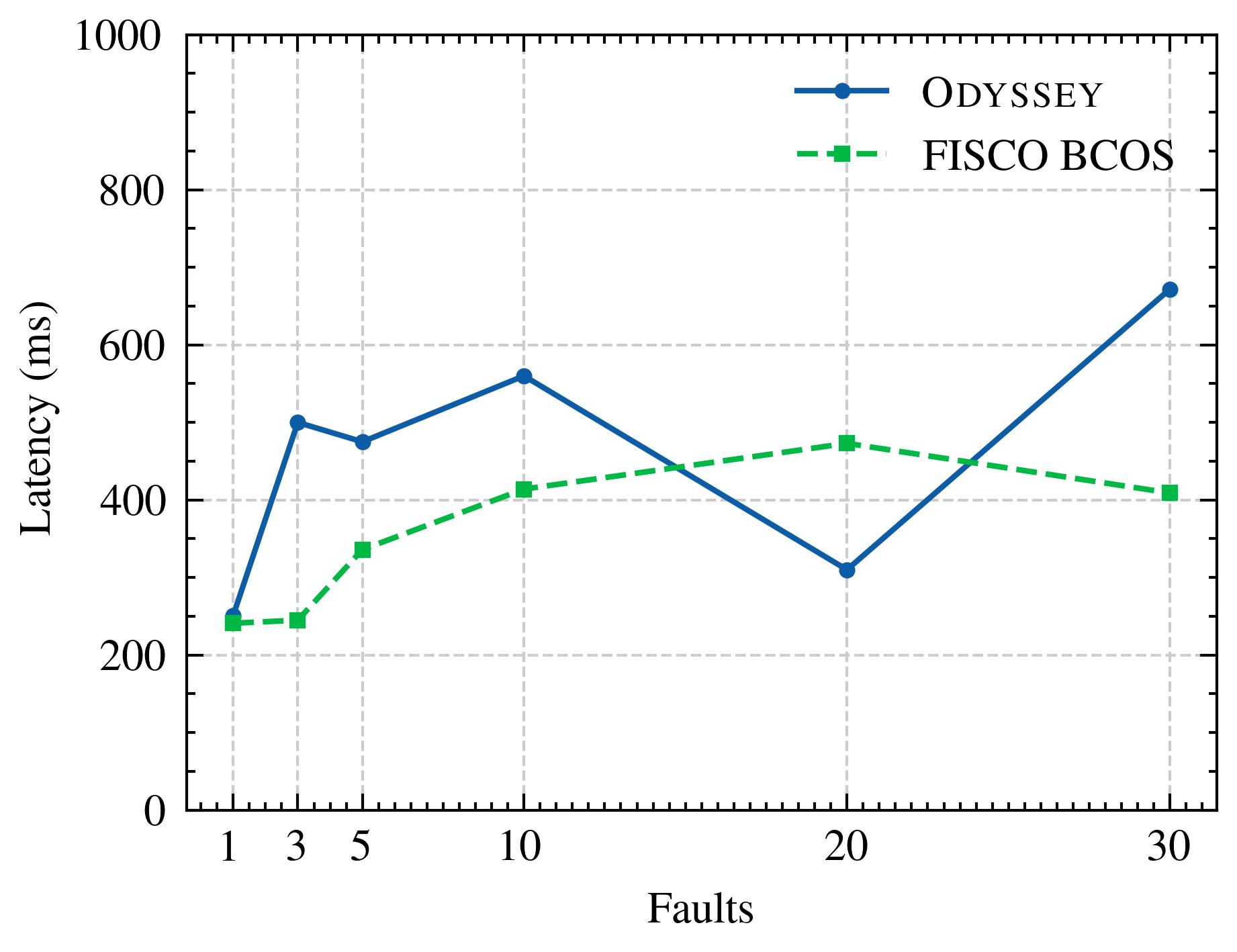}
        \label{lan_max_latency}
    \end{minipage}
}
\caption{Latency vs. throughput in LAN.}
\label{lan test}
\end{figure*} 

\subsubsection{Experiment Setup}
We deploy \sysname and baselines on AWS EC2 machines with one c6a.xlarge instance per node.
All nodes run on dedicated virtual machines with 4 vCPUs and 8GB RAM running Ubuntu Linux 22.04.
We enable SEV-SNP memory encryption for virtual machines.
We send transactions using clients at a specified speed to control the workloads. 
On the client side, we achieve the target transaction sending speed, which is measured in $tx/s$ by controlling the rate at which RPC requests are sent to nodes every second (each RPC request includes one transaction). 

We conduct experiments in both LAN and WAN settings. 
In the LAN setting, we use cloud VM instances in the same availability zone located in the State of Ohio, America, where the latency between instances is 0.2 - 0.5ms.
We use iproute-rc\footnote{https://www.mankier.com/package/iproute-tc} to add a $40 \pm 5$ms latency to the ports of each instance to simulate a WAN environment. 
The round-trip time (RTT) between any two nodes is about 80ms, and 40ms for clients and nodes.
We measure throughput, which denotes the number of transactions the system handles per second (TPS), and latency, which represents the time that each transaction consumes clients send to and get replies from nodes.

\bheading{Benchmark.}
We deploy a transfer scenario including cross-contract parallel transfers for testing \sysname and baseline systems in the cloud environment.
We follow the majority settings of BCOS benchmark.
The benchmark begins by launching 32 accounts with sufficient balances.
Then, we use a client to send a certain number of transactions, calling the smart contract to transfer between different accounts. 
We vary the fault threshold $f \in \{1, 2, 5, 10, 20, 30\}$.
The total number of nodes in \sysname is $2f+1$ and $3f+1$ for BCOS.
Each block has no more than 2000 transactions. 
Furthermore, we limit the minimum seal time of the block package to 50ms.
With this setting, we get better performance for BCOS.
We keep this setting in \sysname for performance comparisons with BCOS.

\subsection{Performance Evaluation} \label{subsec:evaluation}

\subsubsection{LAN Evaluation}
We evaluate \sysname within clusters including 1-, 7-, 11-, 21-, 41-, and 61-nodes in the LAN settings. 
We use clients to deploy and invoke the smart contracts by sending transactions.
We test the throughput of \sysname and BCOS by controlling the speed of sending transactions. 

\bheading{Latency vs. throughput.} 
\figref{lan test} shows the evaluation results. When $f=1$, the maximum throughput of BCOS is about $1 \times$ higher than that of \sysname: \sysname achieves 4973 TPS (with a latency of 251ms) while BCOS can reach 9479 TPS (with a latency of 241ms). 
With $f$ equal to other values, the maximum throughput of BCOS is still higher than \sysname.
\figref{lan test} demonstrates that under different values of $f$, when the throughput of the system reaches a certain value, the latency rises rapidly, \ie, the inflection point in the curve, at which the current throughput can be considered to be the maximum throughput of the system.
\figref{CB_lan} shows that the peak throughput of \sysname does not differ much when $f=1$.
But when $f=2$ or even larger, the peak throughput of \sysname decreases rapidly. 
This is because the state synchronization of \sysname introduces additional network resource usage.
\figref{bcos_lan} shows that the peak throughput of BCOS does not differ much when $f=1, 3, 5$ and $10$, but drops rapidly when $f=20$ and $30$.
By comparing Fig.~\ref{CB_lan} and Fig.~\ref{bcos_lan}, we observe that the latency of \sysname increases faster than that of BCOS when the throughput increases. 

\bheading{Performance with varying number of faults.}
\figref{lan_max_tps} shows that both \sysname and BCOS exhibit obvious performance degradation when the cluster scales up in the LAN. 
For BCOS, the maximum throughput achieved in the 91-node cluster is $4.12\%$ of that in the 4-node cluster. 
As for \sysname, this ratio is $1.35\%$. 
We can see that \sysname achieves poor scaling performance compared to BCOS. 
\figref{lan_max_latency} shows that as the $f$ increases, the latency of \sysname and BCOS generally exhibits an upward trend.
In most cases, the latency of \sysname is higher than that of BCOS.

\subsubsection{WAN Evaluation}

\begin{figure*}[htb]
\centering
\subfigure[\sysname]
{
    \begin{minipage}[c]{.23\textwidth}
        \centering
        \includegraphics[scale=0.55]{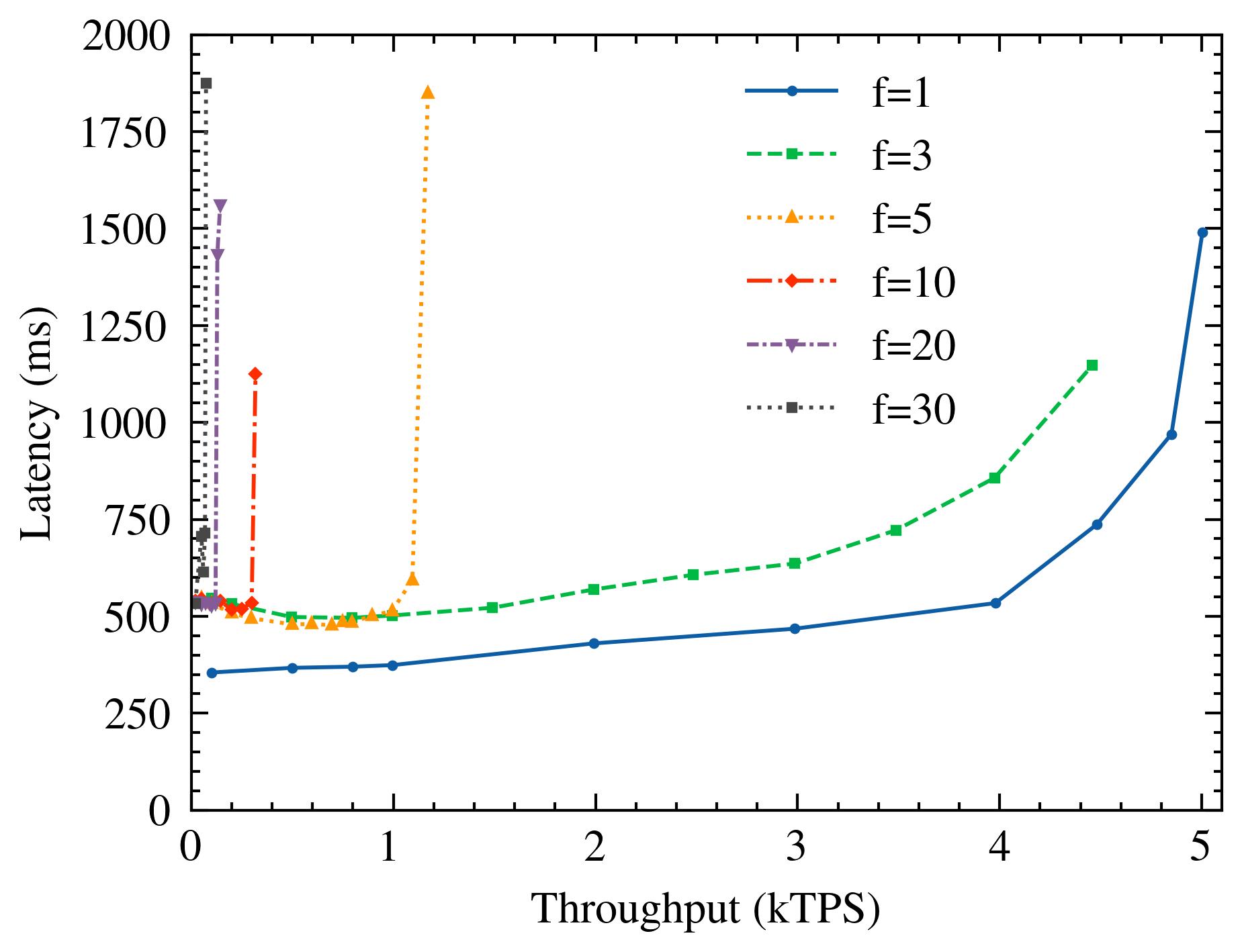}
        \label{CB_wan}
    \end{minipage}
}
\subfigure[FISCO BCOS]
{
 	\begin{minipage}[c]{.23\textwidth}
        \centering
        \includegraphics[scale=0.55]{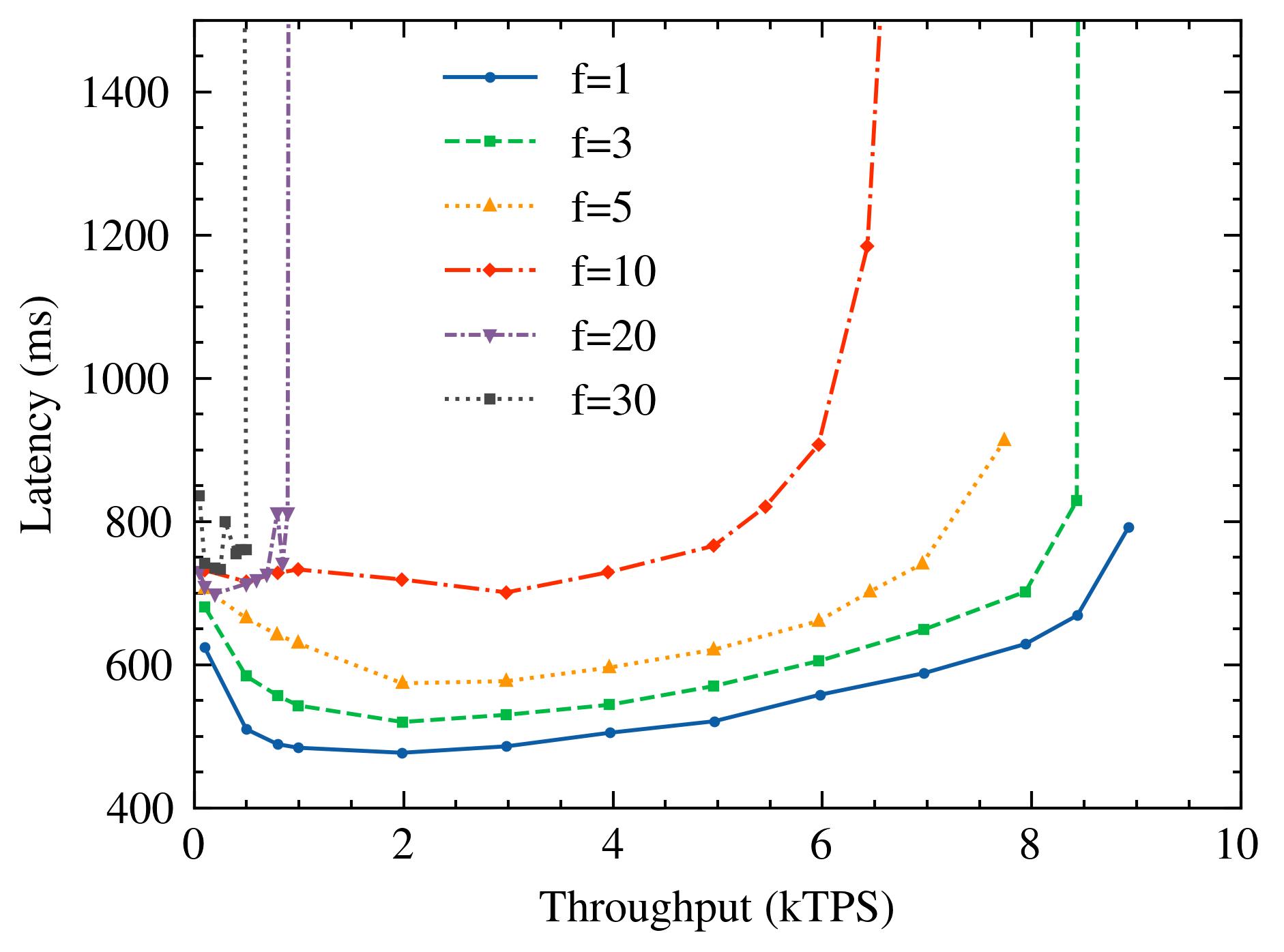}
        \label{bcos_wan}
    \end{minipage}
}
\subfigure[Maximum TPS vs. faults]
{
 	\begin{minipage}[c]{.23\textwidth}
        \centering
        \includegraphics[scale=0.55]{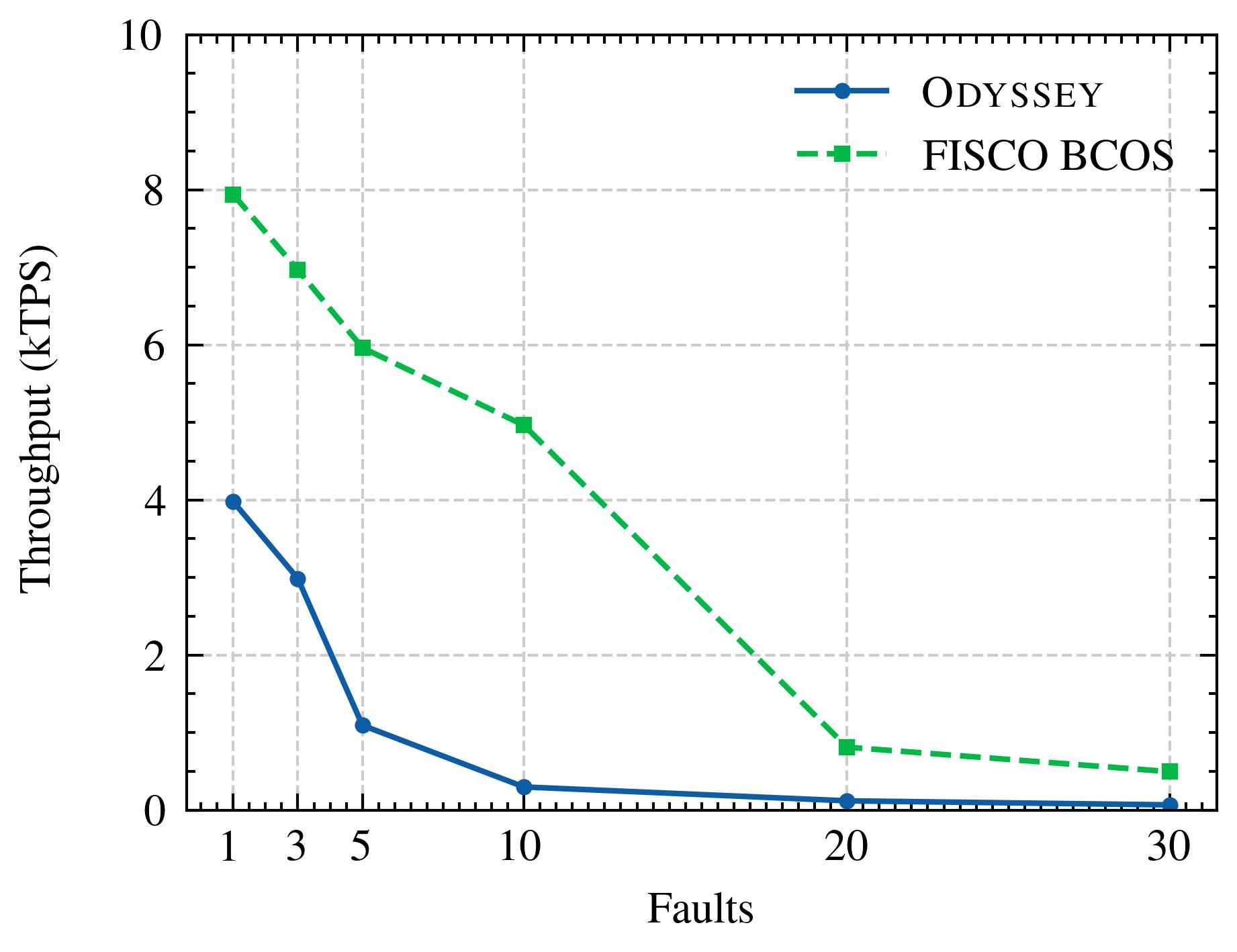}
        \label{wan_max_tps}
    \end{minipage}
}
\subfigure[Maximum latency vs. faults]
{
 	\begin{minipage}[c]{.23\textwidth}
        \centering
        \includegraphics[scale=0.55]{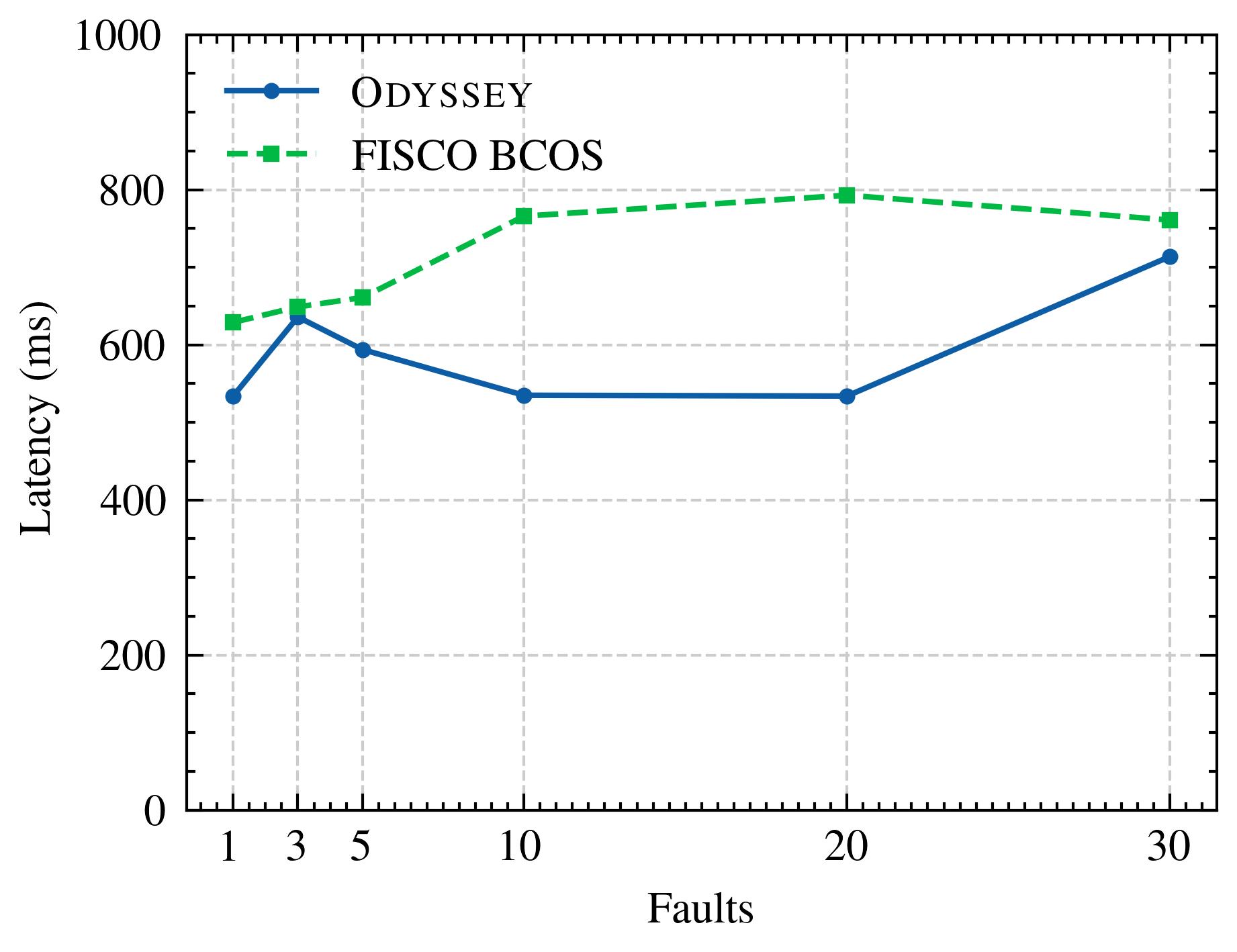}
        \label{wan_max_latency}
    \end{minipage}
}
\caption{Latency vs. throughput in WAN.}
\label{wan test}
\end{figure*}

\bheading{Latency vs. throughput.} 
We evaluate the throughput and latency of \sysname and BCOS in a WAN environment.
The results are shown in \figref{CB_wan}-\figref{wan_max_latency}.
\sysname and BCOS exhibit similar performance trends with those in the LAN environment as shown in \figref{CB_lan}-\figref{lan_max_tps}.
Specifically, \figref{CB_wan} and \figref{bcos_wan} show the latency of \sysname and BCOS as the throughput is increased up to system saturation. 
When $f=1$, the latency of \sysname is low and can be maintained below 500ms before reaching the peak throughput.
As $f$ increases, the latency can hardly go below 500ms.
Moreover, the throughput of \sysname can be up to 4k when $f=1, 2$, and the peak throughput decreases rapidly with the increase of $f$.
From \figref{bcos_wan}, it can be seen that the latency of BCOS is slightly fluctuating, and the peak throughput of BCOS is greater than 6k when $f=1, 3, 5, 10$ and decreases rapidly for larger values of $f$.
By comparing \figref{CB_wan} and \figref{bcos_wan}, we can see that the peak throughput of \sysname is lower than that of BCOS.
When the number of nodes is large ($f=20, 30$), the peak throughputs of both systems are maintained at a low level.
Moreover, the stability of the performance of BCOS is weaker than that of \sysname.

\bheading{Performance with varying number of faults.}  \figref{wan_max_tps} shows the same tendency as shown in \figref{lan_max_tps}. 
BCOS's throughput in the WAN degrades more significantly than that in the LAN. 
Furthermore, \sysname maintains a relatively stable throughput reduction rate. 
However, \sysname demonstrates lower latency than BCOS and shows greater stability than its latency performance in the LAN environment, as shown in \figref{wan_max_latency}.

\bheading{Performance under fault delegator.}
We study the impact on throughput and latency of a crash delegator, with $f=1$ up to the maximum tolerated number of $f=20$ Byzantine nodes ($n=3, 7, 11, 21, 41$ and $61$).
Considering the difference in saturation throughput for different numbers of nodes, we send 100 transactions per second to \sysname, for a total of 10,000 transactions. 
\figref{wan failure} shows the change in throughput and latency as the number of nodes increases. 
\figref{wan_max_tps_failure} shows that when \sysname has 61 nodes, the speed of processing transactions reaches saturation and system throughput drops by 13\%.
When a crash delegator exists, the trend of system throughput is the same as the normal case. 
Specifically, when $n<61$, system throughput is almost unaffected by the crash delegator and can remain stable.
When $n=61$, the presence of the crash delegator causes the system throughput to drop by 24\%, which is more than the normal case. 
\figref{wan_max_latency_failure} shows that when $n<61$, latency keeps stable, while when $n=61$, latency increases by 185\%. 
The crash delegator introduces a significant increase in latency when $n=21, 41$ and $61$, \eg, at $n=41$, latency increases 2.245 seconds compared to the normal case.
As the increase of $n$, the latency variation of \sysname under the impact of the crash delegator increases up to 169\%.

\begin{figure}[htbp]
\centering
\subfigure[Maximum TPS vs. nodes]{
\begin{minipage}[b]{0.45\columnwidth}
    \includegraphics[scale=0.55]{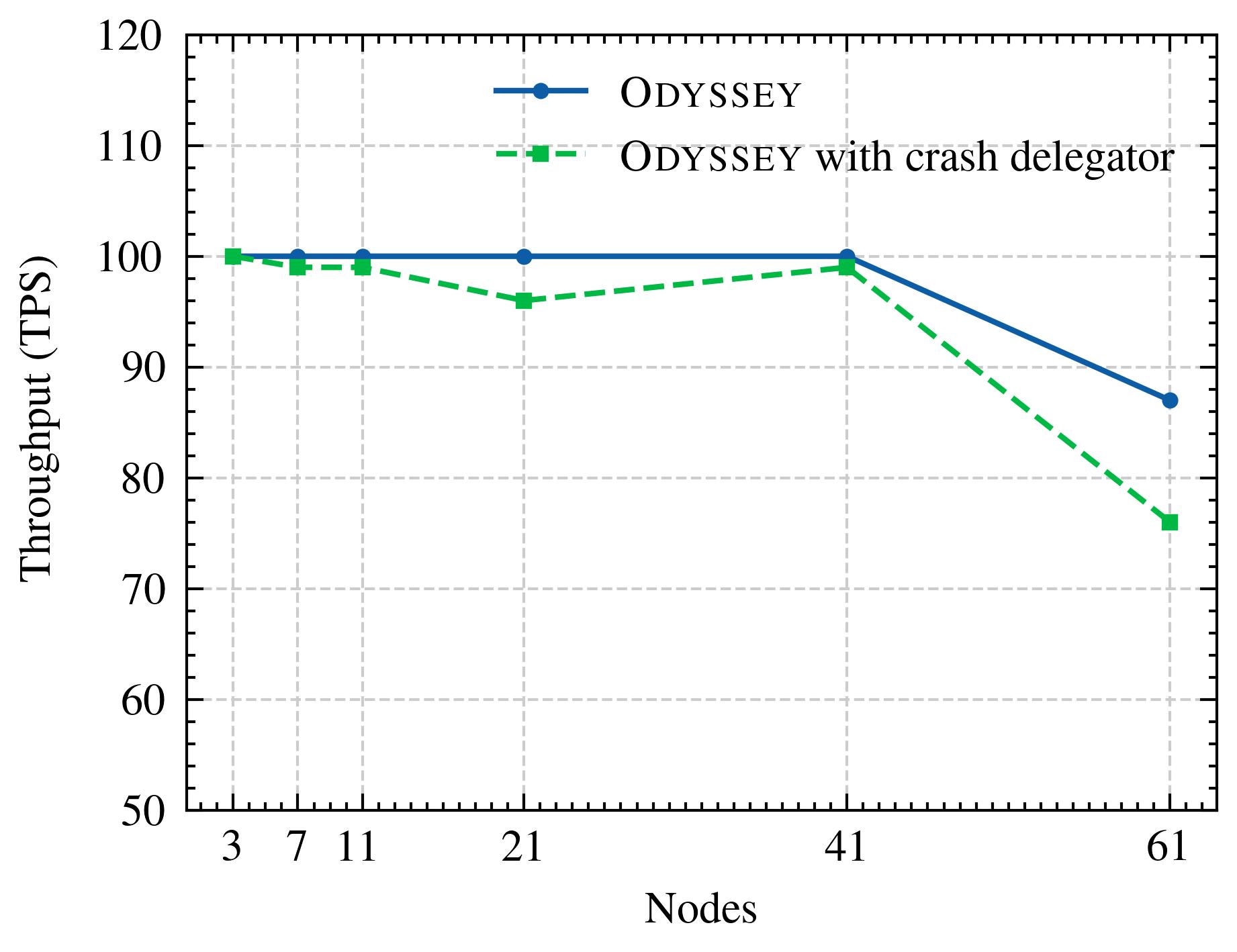}
    \label{wan_max_tps_failure}
\end{minipage}
}
\subfigure[Maximum latency vs. nodes]{
\begin{minipage}[b]{0.45\columnwidth}
    \includegraphics[scale=0.55]{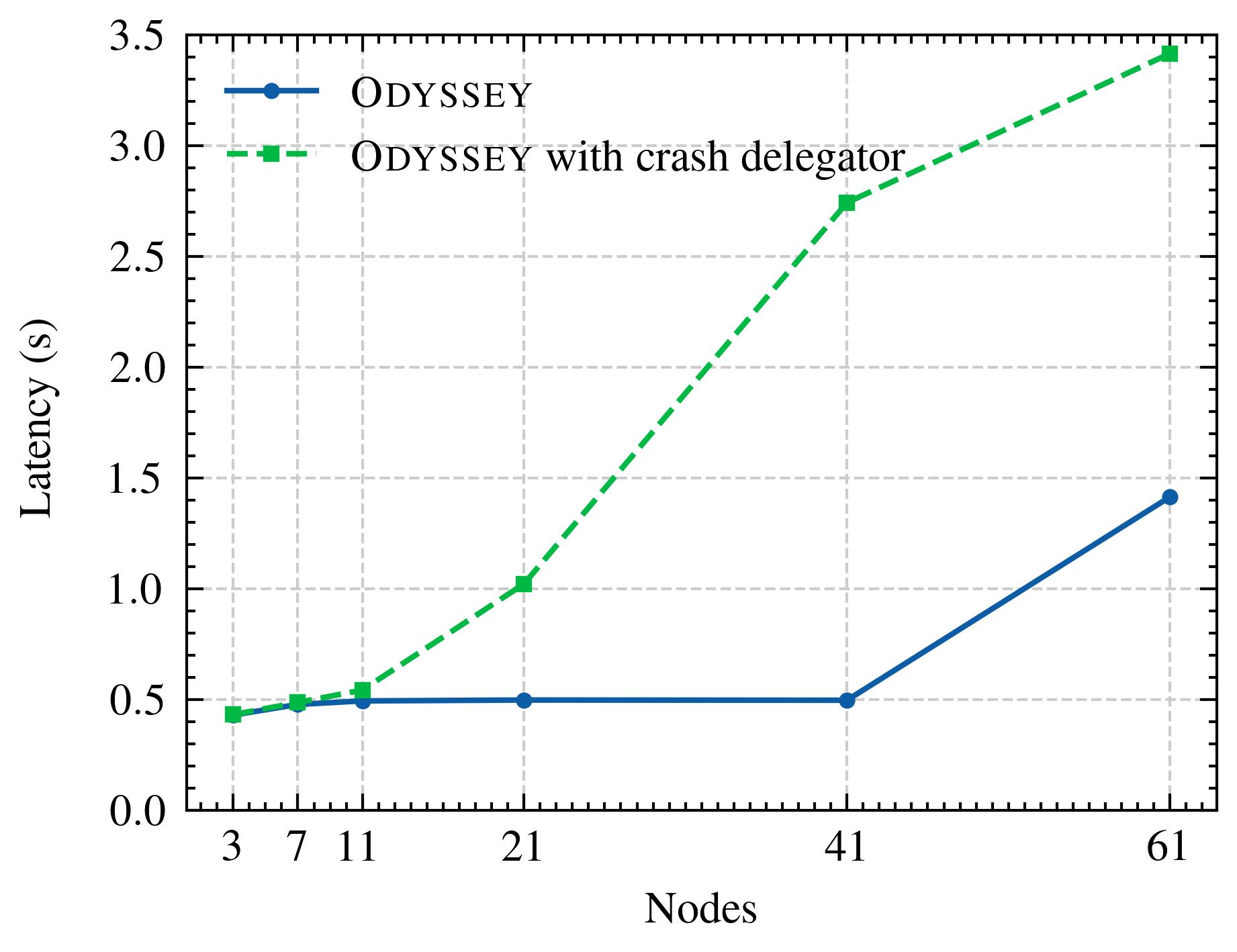}
    \label{wan_max_latency_failure}
\end{minipage}
}

\caption{Latency vs. throughput in WAN with crash delegator.}
\label{wan failure}
\end{figure}

\begin{table}[t]
\centering
\caption{Overhead profiling for \sysname.}
\label{example}
\begin{adjustbox}{width=1\columnwidth,center}
\begin{tabular}{ccc}
\toprule
Protocol & Throughput (TPS) & Latency (ms) \\
\midrule
\sysname & 4850 & 970  \\
\sysname-CS & 4966 & 958  \\
\sysname-ES & 4462 & 820  \\ 
BCOS-S & 5759 & 835  \\ 
BCOS & 8924 & 792  \\
\bottomrule
\end{tabular}
\end{adjustbox}
\label{profile}
\end{table}

\subsection{Overhead Profiling.} \label{subsec:overhead}
To better understand the overhead introduced by \sysname, we compare the performance of \sysname-CS, \sysname-ES, and BCOS-S. 
We construct a 3-node cluster for CFT-based systems (\sysname and \sysname-CS) and a 4-node cluster for BFT-based systems (BCOS and BCOS-S) in the WAN. Table~\ref{profile} lists the throughput and latency of \sysname and its variants. 

\begin{packeditemize}
    \item \textbf{Consensus overhead.} Through the comparison of \sysname and \sysname-ES in the WAN, the gap between \sysname and \sysname-ES is narrow. The peak throughput of \sysname is 1.09 times higher than that of \sysname-ES. It can be seen that the modified consensus layer brings a slight performance improvement for \sysname. 
    Although we employed a VR protocol with only two consensus phases, we added an additional round of communication at the end of the second phase---when the current primary node finalizes the block consensus---to notify other nodes of the consensus results.
    So, the total number of communication rounds is not reduced compared to the PBFT protocol in BCOS.

    \item \textbf{Execution overhead.} The performance difference between \sysname and \sysname-CS suggests that under conditions where delegated execution is used, the performance loss at the execution level is negligible.

    \item \textbf{SEV-SNP overhead.} The peak throughput of BCOS is 1.55 times higher than that of BCOS-S. The use of SEV-SNP confidential virtual machines brings about a large performance loss, and at this point, there is not much difference in the peak throughput between \sysname and BCOS-S. This shows that using an SEV-SNP confidential virtual machine is the main reason for the widening performance gap. 
    Due to the memory encryption operations of the confidential virtual machine and its I/O interactions with external storage devices, greater overhead is introduced.
\end{packeditemize}

\section{Discussion}

\bheading{Lessons learned.}
Compared with traditional methods, our approach has the following benefits. 
First, it does not require software patches, making it feasible for consortium blockchains due to the decentralization characteristic. 
Besides, new side-channel attacks against TEE architectures appear constantly, and so they will exist for a long time, and continuous patches are required.   
Second, it does not introduce significant overhead. 
By contrast, existing protection, such as Raccoon~\cite{rane2015raccoon}, Cloak~\cite{gruss2017strong}, leads to either excessive performance overhead or imperfect side-channel protection.

\bheading{Cloning-based forking attack.}
Wilde \etal~\cite{wilde2025forking} presents an exhaustive analysis of the forking attack that includes rollback-based forking attack, which is discussed in our paper, and cloning-based forking attack.
In cloning-based forking attacks, an adversary exploits the inability of TEEs to distinguish between multiple instances of the same enclave binary running on the same platform. Specifically, it can create multiple identical instances of an enclave on the same machine and choose the preferred results output from TEEs. 
To address this attack, CCF~\cite{howard2023ccf} adopts software-layer solutions. We leave such optimization for \sysname to future work.
\section{Related Work}
\label{sec:related}
In this section, we review prior work on confidential blockchains using TEEs and attacks of TEEs. 

\subsection{Confidential Blockchains}

\iheading{1) L1 solutions.} Brandenburger \etal~\cite{brandenburger2018blockchain} is among the first to port Hyperledger Fabric~\cite{androulaki2018hyperledger} into TEEs to build Confidential blockchain. 
The authors propose application-level barriers to mitigate TEEs' rollback attacks in the execute-then-order architecture.
However, the barriers require much effort to redesign existing blockchain applications (\eg, smart contracts), which is also error-prone. Later, Tz4fabric~\cite{muller2020tz4fabric} only runs the execution of Fabric into TEEs to minimize TCB.
Lew \etal~\cite{lew2024revisiting} revisit the rollback attack issues in Hyperledger Fabric Private Chaincode (FPC) and propose an associated verification mechanism for defense.

Except for Hyperledger Fabric, Microsoft~\cite{russinovich2019ccf, howard2023ccf} realizes an auditable immutable ledger within TEEs to build a general-purpose computing platform, \ie, Confidential Consortium Framework (CCF). CCF allows a crashed TEE node to rejoin the system through reconfiguration to mitigate rollback attacks. 
CONFIDE~\cite{yan2020confidentiality} is another confidential blockchain, which uses a key management protocol to encrypt and decrypt confidential transactions in TEEs.

Most L1 solutions are implemented as consortium blockchains since requiring every participant to possess TEEs is impractical in permissionless blockchains. 
Despite these efforts, studies~\cite{li2024sok, SGXonerated} show that these Confidential blockchains are susceptible to various attacks that compromise confidentiality (\secref{Consortium Blockchains}). This motivates us to propose \sysname, the first Confidential blockchain that comprehensively addresses these attacks. 

\iheading{2) L2 solutions.} L2 solutions~\cite{cheng2019ekiden, tran2018obscuro, lind2019teechain, das2019fastkitten, frassetto2022pose, TeeRollup} mainly offload transaction execution from blockchains (\ie, L1) to off-chain parties' TEEs.  
For example, Ekiden~\cite{cheng2019ekiden}, FastKitten~\cite{das2019fastkitten}, and POSE~\cite{frassetto2022pose} enables off-chain parties to execute smart contracts within their TEEs. 
Specifically, clients can send encrypted transactions to L2 nodes equipped with TEEs where transactions are decrypted and executed.
After that, the results are encrypted and uploaded to L1.
Unlike them, Teechain~\cite{lind2019teechain} enables off-chain nodes to build confidential payment channels to process transactions. 

L2 solutions can be implemented atop both permissionless and permissioned blockchains. Similarly, studies~\cite{li2024sok, SGXonerated} show that these L2 solutions also face confidentiality challenges. While this paper primarily focuses on consortium blockchains, the proposed defenses may also be extended to address these challenges in L2 solutions due to the same confidentiality attacks.

\subsection{Defense against Relevant Attacks}

\bheading{Defenses against rollback attacks.}
To counter rollback attacks against Intel SGX,
Rote~\cite{matetic2017rote} introduces an approach that leverages a distributed system to store enclave-specific counters.
Engraft~\cite{wang2022engraft, tiks} proposes TIKS that uses a KVstore for storing the Raft meta file and the log meta file.
Narrator~\cite{niu2022narrator, narrator-pro} provides state continuity protection for applications running in cloud TEEs. It initiates a distributed system of TEEs by blockchain that can realize fast and unlimited state updates and reads. 
Nimble~\cite{angel2023nimble} introduces an append-only ledger service to guarantee linearizability that can capture liveness with high performance and a small TCB. 
\sysname also addresses the rollback issue against TEEs. 
Unlike the aforementioned methods, \sysname leverages a finalized consensus protocol to fix the transaction execution order, effectively eliminating the possibility of adversaries replaying transactions.

\bheading{Side-channel defenses.}
Side-channel attacks are long-standing threats to TEEs~\cite{li2024sok} and there is no one-size-fits-all solution to eliminate them. Besides, existing defenses, such as Raccoon~\cite{rane2015raccoon}, Cloak~\cite{gruss2017strong}, have excessive performance overhead. 
By contrast, \sysname leverages delegation execution to protect transaction confidentiality in consortium blockchains. 
The design and evaluation showcases that delegation execution is an effective alternative. Extending this execution model to a broader range of side-channel attacks is an interesting direction for future work. 

\section{Conclusion}
This paper identifies \eiatt and \eratt attacks that can breach the confidentiality of existing consortium blockchains. 
To address these attacks, this paper presents \sysname, a Confidential blockchain that leverages the delegated execution scheme and order-then-execute architecture. The scheme allows clients to delegate their transactions to their trustees' TEEs for execution and then synchronize execution results to distrusting participants.
\sysname further leverages two optimizations, \ie, location-aware concurrent execution and delegation failure handler, to reduce delegation overhead.  The extensive evaluation demonstrates that \sysname introduces only slight overhead for the adopted defenses.  

\normalem
\bibliographystyle{IEEEtran}
\bibliography{reference}

\end{document}